\makeatletter
\let\@twosidetrue\@twosidefalse
\let\@mparswitchtrue\@mparswitchfalse
\makeatother

\documentclass{llncs}
\usepackage{amsmath,amsfonts,amssymb}
\usepackage{latexsym}
\usepackage{graphicx}
\usepackage{latexsym,complexity,lmodern}
\usepackage{url}

\usepackage[textsize=tiny,textwidth=2.7cm,shadow]{todonotes}

\newcounter{lpnumber} 
\newtheorem{observation}{Observation}
\newtheorem{new-claim}{Claim}
\newcommand{\vote}{\mathsf{vote}}

\newenvironment{linearprogram}[1]
{ \stepcounter{lpnumber}
  \begin{gather} #1 \tag{LP\arabic{lpnumber}} \\[-5ex] \notag
  \end{gather}
  \hspace{1.5cm} subject to \\[-3ex]
  \align }
{ \endalign }
\newcommand{\minimize}[1]{\text{minimize} \ #1}
\newcommand{\maximize}[1]{\text{maximize} \ #1}

\usepackage{tikz}
\usetikzlibrary{lindenmayersystems,calc,decorations,decorations.pathmorphing,decorations.markings,shapes,shapes.geometric,arrows, backgrounds, patterns,decorations.pathreplacing, positioning}
\pgfdeclarelayer{background}
\pgfdeclarelayer{foreground}
\pgfsetlayers{background,main,foreground}
  
\usetikzlibrary{arrows,decorations.markings,patterns,calc, positioning}
\definecolor{MyPurple}{RGB}{197,0,205}
\pgfmathsetmacro{\d}{1.7}
\pgfmathsetmacro{\b}{2.5}
\tikzstyle{vertex} = [circle, draw=black, fill=black, inner sep=0pt,  minimum size=5pt]
	\tikzstyle{edgelabel} = [circle, fill=white, inner sep=0pt,  minimum size=12pt]
\tikzstyle{arrow} = [line width=0.8mm,-implies,double, double distance=0.8mm]
\tikzstyle{dashedpointedline} = [line width=0.2mm,dashed,dash pattern=on 2mm off 1mm,
decoration={markings,
	mark=at position 1 with {\arrow[line width=0.2mm, scale= 1.8]{>}}
},
postaction={decorate}
] 
\tikzstyle{pointedline} = [line width=0.3mm,
decoration={markings,
	mark=at position 1 with {\arrow[line width=0.2mm, scale= 2]{>}}
},
postaction={decorate}
] 

\tikzstyle{pointedline2} = [line width=0.3mm,
decoration={markings,
	mark=at position 0.75 with {\arrow[line width=0.2mm, scale= 2]{>}}
},
postaction={decorate}
] 

\newcommand{\Xomit}[1]{}
\newcommand{\new}[1]{\textcolor{black}{#1}}

\begin{document}
\title{Popular Matchings in Complete Graphs\thanks{A preliminary version of this work appeared in FSTTCS 2018~\cite{CK18}.}}
\author{\'Agnes Cseh\inst{1}\thanks{Supported by  the Federal
Ministry of Education and Research of Germany in the framework of KI-LAB-ITSE
(project number 01IS19066), the Hungarian Academy of Sciences under its Momentum Programme (LP2016-3/2020), OTKA grant K128611, and  COST Action CA16228 European Network for Game Theory.} \and Telikepalli Kavitha\inst{2}\thanks{This work was done while visiting the Hungarian Academy of Sciences, Budapest.}}
\institute{Centre for Economic and Regional Studies, Institute of Economics, Budapest;\\Hasso Plattner Institute, University of Potsdam, Potsdam; \email{agnes.cseh@hpi.de} \and Tata Institute of Fundamental Research, Mumbai; \email{kavitha@tifr.res.in}}
\maketitle
\pagestyle{plain}

\begin{abstract}
Our input is a complete graph $G$ on $n$ vertices where each vertex has a strict ranking of all other vertices in $G$. The goal is to construct a matching in $G$ that is {\em popular}. A matching $M$ is popular if $M$ does not lose a head-to-head election against any matching $M'$: here each vertex casts a vote for the matching in $\{M,M'\}$ in which it gets a better assignment. Popular matchings need not exist in the  given instance $G$ and the popular matching problem is to decide whether one exists or not. The popular matching problem in $G$ is easy to solve for odd~$n$. Surprisingly, the problem becomes $\NP$-complete for even $n$, as we show here. This is one of the few graph theoretic problems efficiently solvable when $n$ has one parity and $\NP$-complete when $n$ has the other parity.
\end{abstract}

\section{Introduction}
\label{sec:intro}
Consider a complete graph $G = (V,E)$ on $n$ vertices where each vertex ranks all other vertices in a strict order of preference.
Such a graph is called a {\em roommates} instance with complete preferences. The problem of computing a stable matching in $G$ is
classical and well-studied. Recall that a matching $M$ is stable if there is no {\em blocking pair} with respect to $M$, i.e.,
a pair $(u,v)$ where both $u$ and $v$ prefer each other to their respective assignments in $M$.

Stable matchings need not always exist in a roommates instance. For example, the instance
given in Fig.~\ref{fig:nosm} on 4 vertices $d_0,d_1,d_2,d_3$ has no stable matching. (Here $d_0$'s top choice is $d_1$, second choice is $d_2$, and last choice is $d_3$, and similarly for the other vertices.)

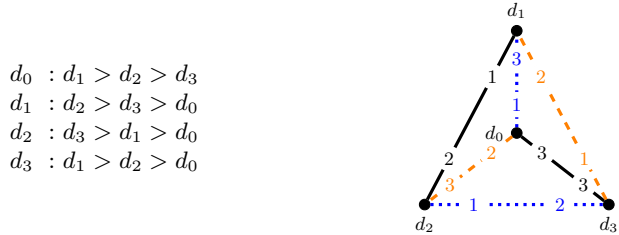
\begin{figure}[h!]
	\centering	
	\begin{minipage}{0.4\textwidth}
		\[
		\begin{array}{ll}
		d_0 \ : & d_1 > d_2 > d_3\\
		d_1 \ : & d_2 > d_3 > d_0\\
		d_2 \ : & d_3 > d_1 > d_0\\
		d_3 \ : & d_1 > d_2 > d_0\\
		\end{array}
		\]
	\end{minipage}\hspace{10mm}\begin{minipage}{0.4\textwidth}
	\begin{tikzpicture}[scale=0.9, transform shape]
	\node[vertex, label=left:$d_0$] (d0) at (0,0) {};
	\node[vertex, label=above:$d_1$] (d1) at (0,\d) {};
	\node[vertex, label=below:$d_2$] (d2) at (-0.9*\d,-0.7*\d) {};
	\node[vertex, label=below:$d_3$] (d3) at (0.9*\d,-0.7*\d) {};
	
	\draw [very thick, blue, dotted] (d0) -- node[edgelabel, near start] {1} node[edgelabel, near end] {3} (d1);
	\draw [very thick, dashed, orange] (d0) -- node[edgelabel, near start] {2} node[edgelabel, near end] {3} (d2);
	\draw [very thick] (d0) -- node[edgelabel, near start] {3} node[edgelabel, near end] {3} (d3);
	
	\draw [very thick] (d1) -- node[edgelabel, near start] {1} node[edgelabel, near end] {2} (d2);
	\draw [very thick, blue, dotted] (d2) -- node[edgelabel, near start] {1} node[edgelabel, near end] {2} (d3);
	\draw [very thick, dashed, orange] (d3) -- node[edgelabel, near start] {1} node[edgelabel, near end] {2} (d1);

	\end{tikzpicture}
\end{minipage}
\caption{An instance that admits two popular matchings---marked by dotted blue and dashed orange edges---but no stable matching. The preference list of each vertex is illustrated by the numbers on its edges: a lower number indicates a more preferred neighbor.}
\label{fig:nosm}
\end{figure}

Irving~\cite{Irv85} gave an efficient algorithm to decide if $G$ admits 
a stable matching or not.
In this paper we consider a notion that is more relaxed than stability: this is the notion of {\em popularity}. 
For any vertex $u$, a ranking over neighbors can be extended naturally to a ranking over matchings as follows:
$u$ prefers matching $M$ to matching $M'$ if (i)~$u$ is matched in $M$ and unmatched in $M'$ or
(ii)~$u$ is matched in both and it prefers its partner in $M$ to its partner in $M'$.
For any two matchings $M$ and $M'$, let $\phi(M,M')$ be the number of vertices that prefer $M$ to $M'$.

\begin{definition}
\label{pop-def}
Let $M$ be any matching in $G$. $M$ is {\em popular} if  $\phi(M,M') \ge  \phi(M',M)$ for every matching $M'$ in $G$.
\end{definition}

Suppose an election is held between $M$ and $M'$ where each vertex casts a vote for the matching that it prefers. So $\phi(M,M')$ (similarly, $\phi(M',M)$) is the number of votes for $M$ (resp., $M'$).  A popular matching $M$ never loses an election to another matching $M'$ since $\phi(M,M') \ge  \phi(M',M)$: thus it is a weak {\em Condorcet winner}~\cite{Con85,wiki-condorcet} in the corresponding
voting instance. So popularity captures collective decision by the vertex set and it can be considered as a natural relaxation
of stability.

The notion of popularity was first introduced in bipartite graphs in 1975 by G\"ardenfors---popular matchings always exist in bipartite graphs since stable matchings always exist here~\cite{GS62} and every stable matching is popular~\cite{Gar75}. The proof that every stable matching is popular holds in non-bipartite graphs as well~\cite{Chu00}; in fact, it is easy to show that every stable matching is a min-size popular matching~\cite{HK11}. Relaxing the constraint of stability to popularity allows us to find feasible matchings that may exist in instances that do not admit stable matchings; moreover, even when stable matchings exist, there may be popular matchings that achieve more ``social good'' (such as larger size), which might be relevant in many applications.

Observe that the instance in Fig.~\ref{fig:nosm} has two popular matchings: $M_1 = \{(d_0,d_1),(d_2,d_3)\}$ and 
$M_2 = \{(d_0,d_2),(d_1,d_3)\}$. However as was the case with stable matchings, popular matchings also need not always exist in the given 
instance $G$. Just take, for example, the same instance as in Fig.~\ref{fig:nosm}, but without vertex~$d_0$. A complete graph $G$
on an even number of vertices that has no popular matching is also easy to describe: take {\em two} copies of this instance on
3 vertices, i.e., $d_1,d_2,d_3$ with preferences as given in Fig.~\ref{fig:nosm} (without vertex~$d_0$) and three more vertices
$d'_1,d'_2,d'_3$ whose preferences are analogous to $d_1,d_2,d_3$, respectively. Since the instance has to be complete, add $d_1,d_2,d_3$
(similarly, $d'_1,d'_2,d'_3$) at the tail of preference lists of $d'_1,d'_2,d'_3$ (resp., $d_1,d_2,d_3$) in some arbitrary order. It is easy
to check that this instance on 6 vertices has no popular matching.

The {\em popular roommates problem} is to decide if $G$ admits a popular
matching or not. When the graph is not complete, it is known that the popular roommates problem is $\NP$-complete~\cite{FKPZ18,GMSZ18}. Here 
we are interested in the complexity of the popular roommates problem when the input instance is complete.

Interestingly, several popular matching problems that are intractable in bipartite graphs become tractable in {\em complete 
bipartite} graphs. The min-cost popular matching problem in bipartite graphs is such a problem---this is $\NP$-hard in a bipartite graph with incomplete lists~\cite{FKPZ18}, however it can be solved in polynomial time in a bipartite graph with complete lists~\cite{CK16}. The difference is due to the fact that while there is no 
compact extended formulation of the convex hull of edge incidence vectors of all popular matchings in a general bipartite graph~\cite{FK20}, this polytope has a compact extended formulation in a complete bipartite graph.

It is a simple observation (see Section~\ref{sec:prelims}) that when $n$ is {\em odd}, a matching in a complete graph $G$ on $n$ vertices is popular only if it is stable. Since there is an efficient algorithm to decide if $G$ admits a stable matching or not, the popular roommates problem in a complete graph $G$ can be efficiently solved when $n$ is odd.
%
%
We show the following result here.

\begin{theorem}
  \label{thm:main}
 Let $G$ be a complete graph on $n$ vertices, where $n$ is even. The problem of deciding whether $G$ admits a popular matching or not is $\NP$-complete.
\end{theorem}

So the popular roommates problem with complete preference lists is $\NP$-complete for even $n$ while it is easy to solve for odd $n$. Some problems possess an inherently different nature depending on the parity of some characteristic input parameter, such as Latin squares~\cite{Jan95} or various problems in voting~\cite{FHS10}. Popular matchings do not belong to this set of problems---note that the popular roommates problem is non-trivial for every $n \ge 5$, i.e., there are both ``yes instances'' and ``no instances'' of size~$n$. It is rare and unusual for a natural decision problem in combinatorial optimization to be efficiently solvable when $n$ has one parity and become $\NP$-complete when $n$ has the other parity. We are not aware of any natural optimization problem on graphs that is non-trivially tractable when the cardinality of the vertex set has one parity, which becomes intractable for the other parity.

\subsection{Background and related work}
The first polynomial time algorithm for the stable roommates problem was given by Irving~\cite{Irv85} in 1985. 
Roommates instances that admit stable matchings were characterized in \cite{Tan91}.
New polynomial time algorithms for the stable roommates problem were given in \cite{Sub94,TS98}.

Algorithmic questions for popular matchings in bipartite graphs  have been well-studied in the last decade~\cite{BIM10,CK16,HK11,HK17,KMN09,Kav12,Kav16}. Not much was known on popular matchings in non-bipartite graphs. Bir\'o et al.~\cite{BIM10} proved that validating whether a given matching is popular can be done in polynomial time, even when ties are present in the preference lists. It was shown in \cite{HK13} that every roommates instance on $n$ vertices admits a matching with {\em unpopularity factor} $O(\log n)$ and that it is $\NP$-hard to compute a least unpopularity factor matching. It was shown in \cite{HK17} that computing a max-weight popular matching in a roommates instance with edge weights is $\NP$-hard, and more recently, that computing a max-size popular matching in a roommates instance is $\NP$-hard~\cite{BK19}.

The complexity of the popular roommates problem was open for several years~\cite{BIM10,Cseh17,HK13,HK17,Man2013} and  two independent $\NP$-completeness proofs~\cite{FKPZ18,GMSZ18} of this problem were shown very recently. Interestingly, both these hardness proofs need ``incomplete preference lists'', i.e., the underlying graph is {\em not} complete. The reduction in \cite{GMSZ18} is from a variant of the vertex cover problem called the {\em partitioned vertex cover} problem and we discuss the reduction in \cite{FKPZ18} in Section~\ref{sec:tech} below. So the complexity status of the popular roommates problem in a complete graph was an open problem and we resolve it here.

An interpretation of roommate instances with complete preference lists might be that each vertex finds every other vertex acceptable, or that being matched to any vertex is better than being unmatched, or that there is no outside option and the agents are all obliged to be matched within the market. Computational hardness for instances with complete lists has been investigated in various matching problems under preferences. An example is the three-sided stable matching problem with cyclic preferences: this involves three groups of participants, say, men, women, and dogs, where dogs have strictly ordered preferences over men only, men have preferences over women only, and finally, women only list the dogs. If these preferences are allowed to be incomplete, the problem of finding a {\em weakly} stable matching is known to be $\NP$-complete~\cite{BM10}. Until very recently, it had been one of the most intriguing open questions in stable matchings~\cite{Man2013,Woe13} as to whether the same problem becomes tractable when lists are complete. Lam and Plaxton gave a hardness proof for complete preference lists very recently, disproving the published conjectures~\cite{LP19}. 


\subsection{Techniques}
\label{sec:tech}
The 1-in-3 SAT problem is a well-known $\NP$-complete problem~\cite{Sch78}: it consists of a Boolean formula $B$ in CNF where every clause has
3 literals (none negated) and the problem is to find a satisfying truth assignment to the variables in $B$ such that every clause has exactly one
literal set to $\mathsf{true}$. We show a polynomial time reduction from 1-in-3 SAT to the popular roommates problem with complete lists.

Our construction is based on the reduction in \cite{FKPZ18} that proved the $\NP$-completeness of the popular roommates problem.
However there are several differences between our reduction and the reduction in \cite{FKPZ18}. The reduction in
\cite{FKPZ18} considered a popular matching problem in bipartite graphs called the ``exclusive popular set'' problem
and showed it to be $\NP$-complete---when preference lists are complete, this problem can be easily solved. Thus
the reduction in \cite{FKPZ18} needs incomplete preference lists.

The exclusive popular set problem asks if there is a popular matching in the given bipartite graph where the set of matched vertices is 
$S$, for a given even-sized subset $S$. A key step in the reduction in~\cite{FKPZ18} from this problem in bipartite graphs to the popular 
matching problem in non-bipartite graphs merges all vertices outside $S$ into a single node. Thus the total number of vertices in the 
non-bipartite graph used in \cite{FKPZ18} is {\em odd}. Moreover, the fact that popular matchings always exist in bipartite graphs 
is crucially used in this reduction. However in our setting, the whole problem is to decide if {\em any} popular 
matching exists in the given graph---thus there are no  popular matchings that ``always exist'' here.

The reduction in \cite{FKPZ18} primarily uses the LP framework of popular matchings in bipartite graphs
from \cite{KMN09,Kav12,Kav18} to analyze the structure of popular matchings in their instance. The LP framework characterizing
popular matchings in non-bipartite graphs is more complex~\cite{Kav18}, so we use the combinatorial characterization of popular 
matchings~\cite{HK11} in terms of forbidden alternating paths/cycles to show that any popular matching in our instance
will yield a 1-in-3 satisfying truth assignment for $B$.
To show the converse, we use a dual certificate similar to the one used in \cite{FKPZ18} to prove the popularity of the matching that we construct using a 1-in-3 satisfying truth assignment for~$B$.

\paragraph{Organization of the paper.} We discuss preliminaries in Section~\ref{sec:prelims}.
Section~\ref{sec:gadgets} describes the construction of our complete graph $G$ corresponding to a given 1-in-3 SAT formula~$B$.
Section~\ref{sec:pop-edges} studies the structure of the graph $G$ and Section~\ref{sec:states} shows that any popular matching in $G$ 
yields a 1-in-3 satisfying truth assignment for~$B$. Section~\ref{se:reverse-construction} completes the reduction by showing how to obtain 
a popular matching in $G$ from any 1-in-3 satisfying truth assignment for~$B$.

\section{Preliminaries}
\label{sec:prelims}

This section contains a characterization of popular matchings from \cite{HK11}. We also include a simple proof of the claim stated in Section~\ref{sec:intro} that when $n$ is odd, every popular matching in $G$ has to be stable.

Let $M$ be any matching in $G = (V,E)$. For any pair $(u,v) \notin M$, define $\vote_u(v,M)$ as follows: (here $M(u)$ is $u$'s partner in $M$
and $M(u) = \mathsf{null}$ if $u$ is unmatched in $M$)
\begin{equation*} 
\vote_u(v,M) = \begin{cases} +   & \text{if\ $u$\ prefers\ $v$\ to\ $M(u)$};\\
	                     - &  \text{if\ $u$\ prefers\ $M(u)$\ to\ $v$.}			
\end{cases}
\end{equation*}

Label every edge $(u,v)$ that does not belong to $M$ by the pair $(\vote_u(v,M),\vote_v(u,M))$. Thus every non-matching edge has a label in $\{(\pm,\pm)\}$. For example, if we consider the matching marked by the dashed orange edges in Fig.~\ref{fig:nosm}, then $(d_1,d_2)$ is labeled $(+,+)$, $(d_2,d_3)$ is labeled $(+,-)$, $(d_0,d_1)$ is labeled $(+,-)$, and $(d_0,d_3)$ is labeled $(-,-)$. Note that an edge is labeled $(+,+)$ if and only if it is a blocking edge to $M$.

We remind the reader that an alternating path/cycle with respect to $M$ is a path/cycle whose alternate edges belong to $M$: thus edges in this path/cycle alternate between belonging to the
matching $M$ and {\em not} belonging to the matching $M$.
Let $G_M$ be the subgraph of $G$ obtained by deleting edges labeled $(-,-)$ from $G$. 
The following theorem characterizes popular matchings in $G$. 

\begin{theorem}[\cite{HK11}]
  \label{thm:char-popular}
$M$ is popular in $G$ if and only if $G_M$ does not contain any of the following with respect to~$M$:
	\begin{enumerate}
		\item[(1)] an alternating cycle with a $(+,+)$ edge;
		\item[(2)] an alternating path with two distinct $(+,+)$ edges;
		\item[(3)] an alternating path with a $(+,+)$ edge and an unmatched vertex as an endpoint of the path.
	\end{enumerate}
\end{theorem}

Using the above characterization, it can be easily checked whether a given matching is popular or not~\cite{HK11}.
Thus our $\NP$-hardness result implies that the popular roommates problem with complete preferences is $\NP$-complete.

\medskip

\noindent{\bf When $n$ is odd.} Recall the claim made in Section~\ref{sec:intro} that when $n$ is odd, every popular matching in $G$
has to be stable. A simple proof of this statement is included below.

\begin{observation}[\cite{huang18}]
Let $G$ be a complete graph on $n$ vertices, where $n$ is odd. Any popular matching in $G$ has to be stable.
\end{observation}

\begin{proof}
 Since $n$ is odd and $G$ is complete, any popular matching leaves exactly one vertex unmatched. Let $M$ be a popular matching and let $v$ be the vertex left unmatched in $M$. Consider a vertex $u$ adjacent to~$v$. We know that $(u,w) \in M$ for some $w \in V \setminus \left\{v \right\}$, and due to Part~(3) in Theorem~\ref{thm:char-popular}, no $(+,+)$ edge is incident to~$w$. Since $v$ is adjacent not only to $u$, but to all vertices in the graph, this holds for all $w \in V$, i.e., there is no $(+,+)$ edge incident to any vertex. Thus $M$ is stable. \qed
\end{proof}

\section{The graph $G$}
\label{sec:gadgets}

Recall that $B$ is the input formula to 1-in-3 SAT. We assume that $B$ has $\kappa$ variables $X_1,\ldots,X_{\kappa}$.
The graph $G$ that we construct here consists of gadgets in 4 levels along with 2 special gadgets that we will call the $D$-gadget and $Z$-gadget.
Gadgets in level~1 correspond to variables in the formula $B$ while gadgets in levels 0, 2, and 3 correspond to clauses in~$B$.
Variants of the gadgets in levels~0-3 and the $D$-gadget were used in \cite{FKPZ18} while the $Z$-gadget is new. 

\medskip

\noindent{\bf An overview.} We will first show that any popular matching $M$ in $G$ uses only {\em intra-gadget} edges. Every gadget will be in 
one of the following two states in $M$: {\em stable} state or {\em unstable} state. A gadget is in stable state if and only if there is no 
edge $e$ with both its endpoints in this gadget such that $e$ blocks $M$. Our aim is to show that for every clause $c$ in $B$,
the gadget of exactly one of the three variables in $c$ is in unstable state in $M$---this will translate to a 1-in-3 satisfying truth
assignment for~$B$.

Preferences will be set such that unstable (similarly, stable) states of gadgets in one level force a certain number 
of gadgets in the adjacent level to be in  unstable (resp., stable) state.
The reduction in \cite{FKPZ18} is also based on the same idea and for each clause in $B$, it used three level~0 gadgets, three level~2 
gadgets, and one level~3 gadget (recall that clause gadgets are in levels~0, 2, 3). In our setting of complete preference lists, we will
need ``duplicates'' or counterparts of all these gadgets to show the above reduction. So for each clause, we will use six level~0 gadgets,
six level~2 gadgets, and two level~3 gadgets. In order to use Theorem~\ref{thm:char-popular}, we will also need a new gadget to ``glue'' alternating
paths across gadgets: this role will be performed by the $Z$-gadget which will be in stable state in $M$.

We will show that every level~3 gadget has to be in unstable state in $M$.
Our technical lemma (Lemma~\ref{lemma2}) proves
that this forces either two of the first three level~2 gadgets or two of the last three level~2 gadgets of every clause to be unstable 
state in $M$. We then show this forces at least one gadget of the three variables in 
every clause in $B$ to be in unstable state in $M$.

We also show that every level~0 gadget has to be in stable state in $M$. This will 
induce at most one gadget of the three variables in every clause in $B$ to be in unstable state in $M$. Thus {\em exactly one} gadget of 
the three variables in every clause in $B$ will be in unstable state in $M$.

\medskip

\noindent{\bf Our gadgets.}
We will now describe all the gadgets that we use here: along with a figure, we provide the preference lists of vertices in this gadget.
The tail of each list consists of all vertices not listed yet, in an arbitrary order. Even though the preference lists are complete, the structure of the gadgets and the preference lists will ensure that inter-gadget edges will not belong to any popular matching, as we will show in Section~\ref{sec:pop-edges}.

\smallskip

%
%
%

\noindent{\bf The $D$-gadget.} The $D$-gadget is on 4 vertices $d_0,d_1,d_2,d_3$ and the preference lists of these vertices are as
given in Fig.~\ref{fig:nosm} with all vertices outside the $D$-gadget at the tail of each list (in an arbitrary order).
Recall that this gadget admits no stable matching. The role of $D$-gadget will be that of a {\em delimiter}---we will show that
every vertex will have to be matched in any popular matching to a neighbor preferred to all its neighbors in the $D$-gadget.

\smallskip

We describe gadgets from level~1 first, then levels~0, 2, 3, and finally, the $Z$-gadget.
The stable matchings within the gadgets are highlighted by colors in the figures.
The gray elements in the preference lists denote vertices that are outside this gadget. 
We will assume that $D$ in a preference list stands for $d_0 > d_1 > d_2 > d_3$. 

\smallskip
\noindent{\bf Level 1.} For each variable $X_i$ in the formula $B$, we construct a gadget on four vertices as shown in Fig.~\ref{fig:1}.

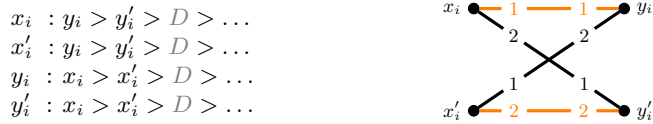
\begin{figure}[h]
	\centering
	
	\begin{minipage}{0.4\textwidth}
		\[
		\begin{array}{ll}
		x_i \ : & y_i > y_i' > \textcolor{gray}{D} > \ldots\\
		x_i' \ : & y_i > y_i' >  \textcolor{gray}{D} > \ldots\\
		y_i \ : & x_i > x_i' >  \textcolor{gray}{D} > \ldots\\
	  y_i' \ : & x_i > x_i' > \textcolor{gray}{D} > \ldots 
		\end{array}
		\]
	\end{minipage}\hspace{10mm}\begin{minipage}{0.4\textwidth}
	\begin{tikzpicture}[scale=0.8, transform shape]
	\node[vertex, label=left:$x_i$] (x) at (0,0) {};
	\node[vertex, label=right:$y_i$] (y) at (\b,0) {};
	\node[vertex, label=right:$y_i'$] (y') at (\b,-\d) {};
	\node[vertex, label=left:$x_i'$] (x') at (0,-\d) {};
	
	\draw [very thick, orange] (x) -- node[edgelabel, near start] {1} node[edgelabel, near end] {1} (y);
	\draw [very thick, orange] (x') -- node[edgelabel, near start] {2} node[edgelabel, near end] {2} (y');
	\draw [very thick] (x') -- node[edgelabel, near start] {1} node[edgelabel, near end] {2} (y);
	\draw [very thick] (x) -- node[edgelabel, near start] {2} node[edgelabel, near end] {1} (y');
	\end{tikzpicture}
\end{minipage}

\caption{The variable gadget in level~1.}
\label{fig:1}
\end{figure}

The bottom vertices $x_i'$ and $y_i'$ will be preferred by some vertices in level~0 to vertices in their own gadget, while the top vertices $x_i$
and $y_i$ will be preferred by some vertices in level~2 to vertices in their own gadget. All four vertices in a level~1 gadget prefer to be matched
among themselves, along the four edges drawn than be matched to any other vertex in the graph. This gadget has a unique stable matching
$\left\{(x_i,y_i),(x'_i,y'_i) \right\}$.

\medskip

\noindent{\bf Level 0.} For each clause $c = X_i \vee X_j \vee X_k $ in the formula $B$, we create 6 gadgets in level~0. For every ordered 
pair of elements in $\left\{i, j, k\right\}$, there is one such gadget. One of these gadgets (this corresponds to the pair $(j,k)$) can be 
seen in Fig.~\ref{fig:0}. The top two vertices, i.e. $a^c_1$ and $b^c_1$, rank $y_j'$ and $x_k'$ in level~1, as their respective second 
choices.
Recall that indices $j$ and $k$ are well-defined in the clause $c = X_i \vee X_j \vee X_k$. Within this level~0 gadget on 
$a^c_1,b^c_1,a^c_2,b^c_2$, both
$\{(a^c_1,b^c_1),(a^c_2,b^c_2)\}$ and $\{(a^c_1,b^c_2),(a^c_2,b^c_1)\}$ are stable matchings.
In the preference lists below (and also for gadgets
in levels~2 and 3), we have omitted the superscript $c$ in their lists for the sake of readability.
\begin{figure}[h]
	\centering
	\begin{minipage}{0.33\textwidth}
		\[
		\begin{array}{ll}
		a_1 \ : & b_1 > \textcolor{gray}{y_j'} > b_2 > \textcolor{gray}{D} > \ldots\\
		a_2 \ : & b_2 > b_1 >  \textcolor{gray}{D} > \ldots\\
		b_1 \ : & a_2 > \textcolor{gray}{x_k'} > a_1 > \textcolor{gray}{D} > \ldots\\
	        b_2 \ : & a_1 > a_2 > \textcolor{gray}{D} > \ldots 
		\end{array}
		\]
	\end{minipage}\begin{minipage}{0.25\textwidth}
	\begin{tikzpicture}[scale=0.8, transform shape]
	\node[vertex, label=left:$a^c_1$] (x) at (0,0) {};
	\node[vertex, label=right:$b^c_1$] (y) at (\b,0) {};
	\node[vertex, label=right:$b^c_2$] (y') at (\b,-\d) {};
	\node[vertex, label=left:$a^c_2$] (x') at (0,-\d) {};
	
	\draw [very thick, orange] (x) -- node[edgelabel, near start] {1} node[edgelabel, near end] {3} (y);
	\draw [very thick, orange] (x') -- node[edgelabel, near start] {1} node[edgelabel, near end] {2} (y');
	\draw [very thick, blue] (x') -- node[edgelabel, near start] {2} node[edgelabel, near end] {1} (y);
	\draw [very thick, blue] (x) -- node[edgelabel, near start] {3} node[edgelabel, near end] {1} (y');
	\end{tikzpicture}
\end{minipage}\begin{minipage}{0.33\textwidth}
		\[
		\begin{array}{ll}
		a'_1 \ : & b'_1 > \textcolor{gray}{y_k'} > b'_2 > \textcolor{gray}{D} > \ldots\\
		a'_2 \ : & b'_2 > b'_1 >  \textcolor{gray}{D} > \ldots\\
		b'_1 \ : & a'_2 > \textcolor{gray}{x_j'} > a'_1 > \textcolor{gray}{D} > \ldots\\
	  b'_2 \ : & a'_1 > a'_2 > \textcolor{gray}{D} > \ldots 
		\end{array}
		\]
	\end{minipage}
\caption{A clause gadget in level~0. The set of preference lists on the left belongs to the first gadget, while the set of preference lists on the right belongs to the fourth gadget (this corresponds to the pair $(k,j)$).}
\label{fig:0}
\end{figure}
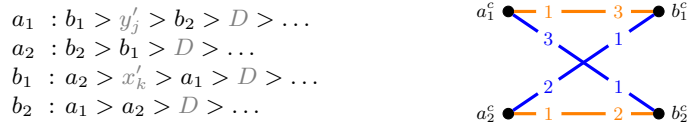

The gadget on vertices $\left\{a^c_3,a^c_4,b^c_3,b^c_4\right\}$ is built analogously: the vertex $a^c_3$ ranks $y_k'$ as its second choice, while $b^c_3$ ranks $x_i'$ second. In the third gadget, the vertex $a^c_5$ ranks $y_i'$ second, while $b^c_5$ ranks $x_j'$ second. Observe the shift in $i,j,k$ indices as
second choices for vertices $a^c_1,a^c_3,a^c_5$ (and similarly, for $b^c_1,b^c_3,b^c_5$).

The fourth, fifth, and sixth gadgets are analogous to their counterparts, the first, second, and third gadgets, respectively, but there is a slight twist.
More precisely, the preferences of $a'^c_1,a'^c_2,b'^c_1,b'^c_2$ in the fourth gadget are analogous to the preferences in Fig.~\ref{fig:0}, except that $a'^c_1$ ranks $y_k'$ second, while $b'^c_1$ ranks $x_j'$ second. Similarly, the second choice of $a'^c_3$ is $y_i'$, the second choice of $b'^c_3$ is $x_k'$, and finally, $a'^c_5$ ranks $y_j'$ second, while $b'^c_5$ ranks $x_i'$ second. Observe the change in {\em orientation} of the indices $i,j,k$ as second choice neighbors when comparing the first three level~0 gadgets of $c$ with its last three level~0 gadgets. This will be important to us later.

\medskip

\noindent{\bf Level 2.} For each clause $c = X_i \vee X_j \vee X_k$ in the formula $B$, we create 6 gadgets in level~2.
The first gadget in level~2 is on vertices $p^c_0,p^c_1,p^c_2,q^c_0,q^c_1,q^c_2$ and their preference lists are described in Fig.~\ref{fig:2}.
Note that $p^c_2$ ranks $y_j$ from level~1 as its second choice, while $q^c_2$ ranks $x_k$ from level~1 second.

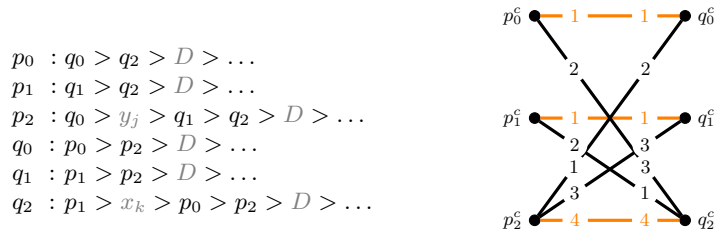
\begin{figure}[h]
	\centering
	\begin{minipage}{0.4\textwidth}
		\[
		\begin{array}{ll}
		p_0 \ : & q_0 > q_2 > \textcolor{gray}{D} > \ldots\\
		p_1 \ : & q_1 > q_2 > \textcolor{gray}{D} > \ldots\\
		p_2 \ : & q_0 > \textcolor{gray}{y_j} > q_1 > q_2 > \textcolor{gray}{D} > \ldots\\
                q_0 \ : & p_0 > p_2 > \textcolor{gray}{D} > \ldots\\
		q_1 \ : & p_1 > p_2 > \textcolor{gray}{D} > \ldots\\
		q_2 \ : & p_1 > \textcolor{gray}{x_k} > p_0 > p_2 > \textcolor{gray}{D} > \ldots\\
		\end{array}
		\]
	\end{minipage}\hspace{10mm}\begin{minipage}{0.4\textwidth}
	\begin{tikzpicture}[scale=0.8, transform shape]
	\node[vertex, label=left:$p^c_0$] (p0) at (0,0) {};
	\node[vertex, label=right:$q^c_0$] (q0) at (\b,0) {};
	\node[vertex, label=left:$p^c_1$] (p1) at (0,-\d) {};
	\node[vertex, label=right:$q^c_1$] (q1) at (\b,-\d) {};
	\node[vertex, label=left:$p^c_2$] (p2) at (0,-2*\d) {};
	\node[vertex, label=right:$q^c_2$] (q2) at (\b,-2*\d) {};
	
	\draw [very thick, orange] (p0) -- node[edgelabel, near start] {1} node[edgelabel, near end] {1} (q0);
	\draw [very thick, orange] (p1) -- node[edgelabel, near start] {1} node[edgelabel, near end] {1} (q1);
	\draw [very thick, orange] (p2) -- node[edgelabel, near start] {4} node[edgelabel, near end] {4} (q2);
	\draw [very thick] (p0) -- node[edgelabel, near start] {2} node[edgelabel, near end] {3} (q2);
	\draw [very thick] (p2) -- node[edgelabel, near start] {1} node[edgelabel, near end] {2} (q0);
	\draw [very thick] (p1) -- node[edgelabel, near start] {2} node[edgelabel, near end] {1} (q2);
	\draw [very thick] (p2) -- node[edgelabel, near start] {3} node[edgelabel, near end] {2} (q1);
	
	\end{tikzpicture}
\end{minipage}

\caption{A clause gadget in level~2.}
\label{fig:2}
\end{figure}

The second gadget in level~2 is on vertices $p^c_3,p^c_4,p^c_5,q^c_3,q^c_4,q^c_5$ and it is built analogously. That is, $p^c_3$ and $q^c_3$ are each
other's top choices and similarly, $p^c_4$ and $q^c_4$ are each other's top choices, and so on. The preference list of $p^c_5$ is
$q^c_3 > \textcolor{gray}{y_k} > q^c_4 > q^c_5 > \textcolor{gray}{D} > \ldots$ and the preference list of $q^c_5$ is
$p^c_4 > \textcolor{gray}{x_i} > p^c_3 > p^c_5 > \textcolor{gray}{D} > \ldots$

The third gadget in level~2 is on vertices $p^c_6,p^c_7,p^c_8,q^c_6,q^c_7,q^c_8$ and it is built analogously. In particular,
the preference list of $p^c_8$ is $q^c_6 > \textcolor{gray}{y_i} > q^c_7 > q^c_8 > \textcolor{gray}{D} > \ldots$ and the
preference list of $q^c_8$ is $p^c_7 > \textcolor{gray}{x_j} > p^c_6 > p^c_8 > \textcolor{gray}{D} > \ldots$

The fourth gadget in level~2 is on vertices $p'^c_0,p'^c_1,p'^c_2,q'^c_0,q'^c_1,q'^c_2$ and it is totally analogous to its counterpart, the first gadget in level~2. That is, $p'^c_0$ and $q'^c_0$ are each other's top choices and similarly, $p'^c_1$ and $q'^c_1$ are each other's top choices, and so on. In particular, the preference list of $p'^c_2$ is
$q'^c_0 > \textcolor{gray}{y_j} > q'^c_1 > q'^c_2 > \textcolor{gray}{D} > \ldots$ and the preference list of $q'^c_2$ is
$p'^c_1 > \textcolor{gray}{x_k} > p'^c_0 > p'^c_2 > \textcolor{gray}{D} > \ldots$

Similarly, the fifth gadget in level~2 is on vertices $p'^c_3,p'^c_4,p'^c_5,q'^c_3,q'^c_4,q'^c_5$ and it is totally analogous to the second gadget in level~2. Also, the sixth gadget in level~2 is on vertices $p'^c_6,p'^c_7,p'^c_8,q'^c_6,q'^c_7,q'^c_8$ and it is totally analogous to the third gadget in level~2.

\medskip

\noindent{\bf Level 3.} For each clause $c = X_i \vee X_j \vee X_k $ in the formula $B$, we create 2 gadgets in level~3.
The first gadget is on vertices $s^c_0,s^c_1,s^c_2,s^c_3,t^c_0,t^c_1,t^c_2,t^c_3$ and the preference lists of these vertices are described in Fig.~\ref{fig:3}.

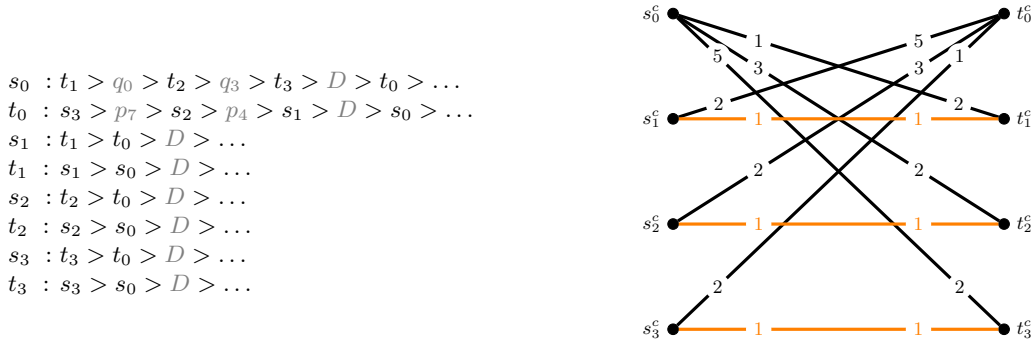
\begin{figure}[h]
	\centering
	\begin{minipage}{0.4\textwidth}
		\[
                \begin{array}{ll}
		s_0 \ : & t_1 > \textcolor{gray}{q_0} > t_2 > \textcolor{gray}{q_3} > t_3 > \textcolor{gray}{D} > t_0 > \ldots\\
		t_0 \ : & s_3 > \textcolor{gray}{p_7} > s_2 > \textcolor{gray}{p_4} > s_1 > \textcolor{gray}{D} > s_0 >  \ldots\\
		s_1 \ : & t_1 > t_0 > \textcolor{gray}{D} > \ldots\\
		t_1 \ : & s_1 > s_0 > \textcolor{gray}{D} > \ldots\\
		s_2 \ : & t_2 > t_0 > \textcolor{gray}{D} > \ldots\\
		t_2 \ : & s_2 > s_0 > \textcolor{gray}{D} > \ldots\\
		s_3 \ : & t_3 > t_0 > \textcolor{gray}{D} > \ldots\\
		t_3 \ : & s_3 > s_0 > \textcolor{gray}{D} > \ldots\\
		
		\end{array}
		\]
	\end{minipage}\hspace{23mm}\begin{minipage}{0.3\textwidth}
	\begin{tikzpicture}[scale=0.8, transform shape]
	\pgfmathsetmacro{\d}{1.75}
\pgfmathsetmacro{\b}{5.5}

	\node[vertex, label=left:$s^c_0$] (s0) at (0,0) {};
	\node[vertex, label=right:$t^c_0$] (t0) at (\b,0) {};
	\node[vertex, label=left:$s^c_1$] (s1) at (0,-\d) {};
	\node[vertex, label=right:$t^c_1$] (t1) at (\b,-\d) {};
	\node[vertex, label=left:$s^c_2$] (s2) at (0,-2*\d) {};
	\node[vertex, label=right:$t^c_2$] (t2) at (\b,-2*\d) {};
	\node[vertex, label=left:$s^c_3$] (s3) at (0,-3*\d) {};
	\node[vertex, label=right:$t^c_3$] (t3) at (\b,-3*\d) {};
	
	\draw [very thick] (s0) -- node[edgelabel, near start] {1} node[edgelabel, very near end] {2} (t1);
	\draw [very thick] (s0) -- node[edgelabel, near start] {3} node[edgelabel, near end] {2} (t2);
	\draw [very thick] (s0) -- node[edgelabel, very near start] {5} node[edgelabel, very near end] {2} (t3);
	\draw [very thick] (t0) -- node[edgelabel, near start] {5} node[edgelabel, very near end] {2} (s1);
	\draw [very thick] (t0) -- node[edgelabel, near start] {3} node[edgelabel, near end] {2} (s2);
	\draw [very thick] (t0) -- node[edgelabel, very near start] {1} node[edgelabel, very near end] {2} (s3);
	
	\draw [very thick, orange] (t1) -- node[edgelabel, near start] {1} node[edgelabel, near end] {1} (s1);
	\draw [very thick, orange] (t2) -- node[edgelabel, near start] {1} node[edgelabel, near end] {1} (s2);
	\draw [very thick, orange] (t3) -- node[edgelabel, near start] {1} node[edgelabel, near end] {1} (s3);
	
	\end{tikzpicture}
\end{minipage}
\caption{A clause gadget in level~3.}
\label{fig:3}
\end{figure}

The counterpart of the first gadget in level~3 is the second gadget in level~3. It is on vertices $s'^c_0,s'^c_1,s'^c_2,s'^c_3,t'^c_0,t'^c_1,t'^c_2,t'^c_3$ and their preference lists are totally analogous to the preference lists of the first gadget in level~3. 

\medskip

\noindent{\bf The $Z$-gadget.} The $Z$-gadget is on 6 vertices $z_0,z_1,z_2,z_3,z_4,z_5$ and the preference lists of these vertices are given in Fig.~\ref{fig:z}. The vertices in a set stand for all these vertices in a fixed arbitrary order. For example, $\cup_{i=1}^{\kappa}{\left\{x_i,y_i\right\}}$ denotes all the ``top'' vertices belonging to 
all the $\kappa$ variable gadgets in a fixed arbitrary order.

\begin{figure}[h]
	\centering
	\begin{minipage}{0.47\textwidth}
		\[
		\begin{array}{ll}
		z_0 \ : & z_4 > z_5 > \textcolor{gray}{\cup_{i=1}^{\kappa}\left\{x_i,y_i\right\}} >  \textcolor{gray}{{\cup_{i=0}^2\cup_c}\left\{p^c_{3i+1},q^c_{3i},p'^c_{3i+1},q'^c_{3i}\right\}} > \\
		&\textcolor{gray}{{\cup_{i=1}^6\cup_c}\left\{a^c_i,b^c_i,a'^c_i,b'^c_i\right\}} > z_1 > z_2 > z_3 > \textcolor{gray}{D} > \ldots\\
		z_1 \ : & z_5 > z_4 > \textcolor{gray}{\cup_{i=1}^{\kappa}\left\{x_i,y_i\right\}} >  \textcolor{gray}{\cup_{i=0}^2\cup_c\left\{p^c_{3i+1},q^c_{3i},p'^c_{3i+1},q'^c_{3i}\right\}} > \\
		&\textcolor{gray}{\cup_{i=1}^6\cup_c\left\{a^c_i,b^c_i,a'^c_i,b'^c_i\right\}} > z_0 > z_3 > z_2 > \textcolor{gray}{D} > \ldots\\
		z_2 \ : & z_0 > z_1> z_3 >z_4 >z_5 > \textcolor{gray}{D} > \ldots\\
		z_3 \ : & z_1 > z_0> z_2 >z_5 >z_4 > \textcolor{gray}{D} > \ldots\\
		z_4 \ : & z_2 > z_3> z_5 >z_0 >z_1 > \textcolor{gray}{D} > \ldots\\
		z_5 \ : & z_3 > z_2> z_4 >z_1 >z_0 > \textcolor{gray}{D} > \ldots\\
		\end{array}
		\]
	\end{minipage}\hspace{25mm}\begin{minipage}{0.3\textwidth}
	\begin{tikzpicture}[scale=0.8, transform shape]
	\pgfmathsetmacro{\d}{2.75}
\pgfmathsetmacro{\b}{5.5}
	\node[vertex, label=left:$z_0$] (z0) at (0,0) {};
	\node[vertex, label=right:$z_1$] (z1) at (\b,0) {};
	\node[vertex, label=left:$z_2$] (z2) at (0.25*\b,-\d) {};
	\node[vertex, label=right:$z_3$] (z3) at (0.75*\b,-\d) {};
	\node[vertex, label=left:$z_4$] (z4) at (0,-2*\d) {};
	\node[vertex, label=right:$z_5$] (z5) at (\b,-2*\d) {};
	
	\draw [very thick, orange] (z0) -- node[edgelabel, near start] {3} node[edgelabel, near end] {3} (z1);
	\draw [very thick, orange] (z2) -- node[edgelabel, near start] {3} node[edgelabel, near end] {3} (z3);
	\draw [very thick, orange] (z4) -- node[edgelabel, near start] {3} node[edgelabel, near end] {3} (z5);
	
	\draw [very thick, gray, dotted] (z0) -- node[edgelabel, near start] {2} node[edgelabel, near end] {5} (z5);
	\draw [very thick, gray, dotted] (z2) -- node[edgelabel, near start] {2} node[edgelabel, near end] {5} (z1);
	\draw [very thick, gray, dotted] (z4) -- node[edgelabel, near start] {2} node[edgelabel, near end] {5} (z3);
	\draw [very thick, gray, dotted] (z1) -- node[edgelabel, near start] {2} node[edgelabel, near end] {5} (z4);
	\draw [very thick, gray, dotted] (z3) -- node[edgelabel, near start] {2} node[edgelabel, near end] {5} (z0);
	\draw [very thick, gray, dotted] (z5) -- node[edgelabel, near start] {2} node[edgelabel, near end] {5} (z2);
	
	\draw [very thick] (z0) -- node[edgelabel, near start] {1} node[edgelabel, near end] {4} (z4);
	\draw [very thick] (z4) -- node[edgelabel, near start] {1} node[edgelabel, near end] {4} (z2);
	\draw [very thick] (z2) -- node[edgelabel, near start] {1} node[edgelabel, near end] {4} (z0);
	\draw [very thick] (z1) -- node[edgelabel, near start] {1} node[edgelabel, near end] {4} (z5);
	\draw [very thick] (z5) -- node[edgelabel, near start] {1} node[edgelabel, near end] {4} (z3);
	\draw [very thick] (z3) -- node[edgelabel, near start] {1} node[edgelabel, near end] {4} (z1);
	
	\end{tikzpicture}
\end{minipage}
\caption{The $Z$-gadget.}
\label{fig:z}
\end{figure}
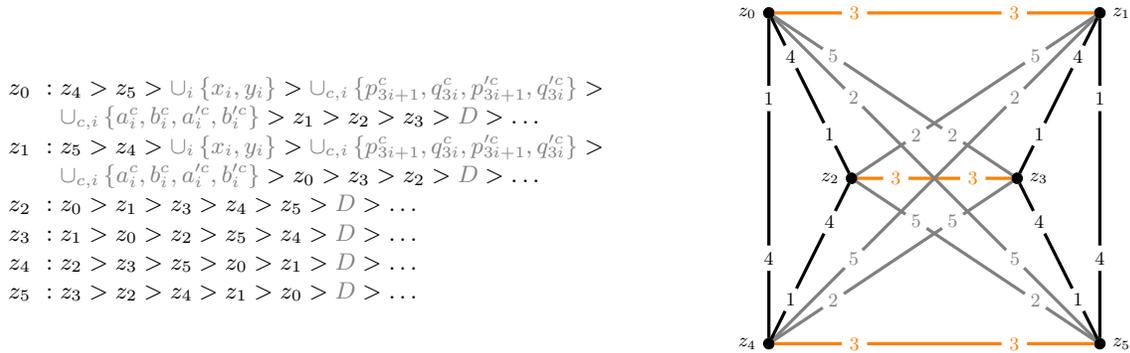

Note that $G$ is a complete graph on an even number of vertices and so every popular matching in $G$ has to be a perfect matching.


\section{Popular edges in $G$}
\label{sec:pop-edges}
%
%
%
Call an edge $e$ in $G$ {\em popular} if there is a popular matching $M$ in $G$ such that $e \in M$.
In this section we identify edges that cannot be popular and show that every popular edge has to be an {\em intra-gadget} edge,  i.e., it connects two vertices of the same gadget.


The following observation, which is straightforward, will be used repeatedly in our proofs.

\begin{observation}
  \label{obs:first_choice}
  Let $v$ be $u$'s top choice neighbor. If $v$ is matched in $M$ to a neighbor worse than $u$ then $(u,v)$ is a blocking edge to $M$.
\end{observation}

\new{We now start restricting the set of edges that can possibly occur in a popular matching. Our first lemma eliminates some of the inter-gadget edges incident to vertices $s^c_0$, $t^c_0$, $s'^c_0$, and $t'^c_0$ in level~3.}

\begin{lemma}
  \label{first-lemma}
  For any clause $c$, no popular matching in $G$ can match $s^c_0$ (similarly, $t^c_0$) to a neighbor worse than $t^c_0$ (resp., $s^c_0$).
  An analogous statement holds for $s'^c_0$ and $t'^c_0$.
\end{lemma}

\begin{proof}
  Let $M$ be a popular matching such that
  $(s^c_0,v) \in M$ for some vertex $v$ such that $t^c_0 > v$ in $s^c_0$'s list, i.e., $s^c_0$ prefers $t^c_0$ to $v$. We claim this implies:
  \begin{itemize}
  \item  a $(+,+)$ edge reachable from $v$ via an alternating path in $G_M$ that begins with a non-matching edge incident to $v$ {\em and}
  \item  a $(+,+)$ edge reachable from $s^c_0$ via an alternating path in $G_M$ that begins with a non-matching edge incident to $s^c_0$.
  \end{itemize}

  If this is the same $(+,+)$ edge then we have an alternating cycle in $G_M$ with a  $(+,+)$ edge, a contradiction to $M$'s popularity
  (by Theorem~\ref{thm:char-popular}). If these are two different $(+,+)$ edges then there is an alternating path in $G_M$ with  two $(+,+)$ edges,
  again a contradiction to $M$'s popularity (by Theorem~\ref{thm:char-popular}).

  \begin{itemize}
  \item[(1)] If $v$ is a top choice neighbor for some vertex (such as $x_i,y_i,z_j$ for  $1 \le i \le \kappa$, $0 \leq j \leq 5$ or $d_1,d_2,d_3$ or $a^r_{\ell},b^r_{\ell},p^r_{3h},p^r_{3h+1},q^r_{3h},q^r_{3h+1}$ or their counterparts for $1 \le \ell \le 6$, $0 \le h \le 2$ and any clause~$r$)
  then there is a $(+,+)$ edge incident to~$v$
    (by Observation~\ref{obs:first_choice}).
  \item[(2)] Suppose $v$ is one of $s^r_0,t^r_0, s'^r_0, t'^r_0$ for some clause $r$ (note that $r \neq c$ in the case of $s^r_0,t^r_0$).
    Assume without loss of generality that $v = s^r_0$. Then either $(s^r_0,t^r_0)$ is a
    $(+,+)$ edge or $t^r_0$ is matched in $M$ to a neighbor better than $s^r_0$. 

    Recall $t^r_0$'s preference list: every vertex that $t^r_0$ prefers to $s^r_0$ is either  a top choice neighbor or it is $d_0$.
    In the former case, there is a $(+,+)$ edge incident to $t^r_0$'s partner (by Observation~\ref{obs:first_choice}) and in the latter case also there is a $(+,+)$ edge
    incident to $d_0$ since one of $d_1, d_2, d_3$ is matched to a neighbor worse than $d_0$ and so there is a $(+,+)$ edge between this vertex in the set $\{d_1, d_2, d_3\}$ and $d_0$.
    Since the edge $(s^r_0,t^r_0)$ is a $(+,-)$ edge, there is a $(+,+)$ edge reachable from $v = s^r_0$ via an alternating path of length~2.
  \item[(3)]  The only case left is when $v$ is neither a top choice neighbor of some vertex nor one of $s^r_0,t^r_0, s'^r_0, t'^r_0$ for 
some clause $r$. So $v$ is a vertex such as $d_0$ or $x'_i,y'_i$ (for $1 \le i \le \kappa$) or 
$p^{r}_{3h+2},q^{r}_{3h+2},p'^{r}_{3h+2},q'^{r}_{3h+2}$ (for $h = 0,1,2$ and some clause $r$).
    It is easy to see that there is a $(+,+)$ edge reachable
    from $v$ via an alternating path of length at most~2. For instance, either $(x'_i,y'_i)$ is a $(+,+)$ edge or $(x_i,y'_i) \in M$ which 
    creates the alternating path $(s^c_0,x'_i)-(y'_i,x_i)-(y_i,\ast)$, where $(x_i,y_i)$ is a $(+,+)$ edge.
  \end{itemize}
  
  Similarly, we can argue that there is a $(+,+)$ edge reachable from $s^c_0$ via an alternating path in~$G_M$. If $t^c_0$ is matched to a
  neighbor worse than $s^c_0$ then the edge $(s^c_0,t^c_0)$ is a $(+,+)$ edge. Else $t^c_0$ is matched to a neighbor $u$ better than $s^c_0$
  and this means there is a $(+,+)$ edge incident to $u$, as we argued above in case~(2).  Hence there is a $(+,+)$ edge reachable from $s^c_0$
  via an alternating path of length at most 2 in $G_M$. \qed
\end{proof}  

\new{Next we show that no inter-gadget edge incident to any vertex in the $D$-gadget can appear in any popular matching.}

\begin{lemma}
  \label{second-lemma}
  Every popular matching matches the vertices in the $D$-gadget among themselves.
\end{lemma}

\begin{proof}
  Let $M$ be a matching that matches $d_i$ for some $i \in \{0,1,2,3\}$ to a vertex $v$ outside the $D$-gadget. This means at least 2 vertices
  $d_i$ and $d_j$ in the $D$-gadget are matched to vertices outside the $D$-gadget. So $(d_i,d_j)$ is a $(+,+)$ edge. We now claim there is a forbidden
  alternating path or cycle (as given in Theorem~\ref{thm:char-popular}) to $M$'s popularity.

  If $v$ is a top choice neighbor or a vertex such as $x'_i,y'_i$ (for $1 \le i \le \kappa$) or $p^c_{3h+2},q^c_{3h+2},p'^c_{3h+2},q'^c_{3h+2}$
  (for $h = 0,1,2$ and some clause $c$) then there is a $(+,+)$ edge $e$
  reachable from $v$ via an alternating path of length at most 2 as seen in the proof of Lemma~\ref{first-lemma}.
  This creates an alternating path in $G_M$ with 2 $(+,+)$ edges: $(d_j,d_i)$ and $e$.
  
  The other possibility is that $v$ is $s^c_0, t^c_0, s'^c_0, t'^c_0$ for some clause~$c$.
  Assume without loss of generality that $v = s^c_0$. Consider the vertex $t^c_0$. We know from Lemma~\ref{first-lemma} that $t^c_0$ has to be matched
  to a neighbor at least as good as $s^c_0$. So we have the following cases:
  \begin{itemize}
  \item[(1)] $t^c_0$ is matched to $d_{j}$ in the $D$-gadget: this means there is either an alternating path with 2 $(+,+)$ edges or an alternating cycle
    with a $(+,+)$ edge:
    \[ (s^c_0,d_i) \overset{(+,+)}{-} (d_{j},t^c_0) \overset{(+,+)}{-} (s^c_1,\ast) \ \ \ \ \ \ \ \mathrm{or}\ \ \ \ \ \ \ (s^c_0,d_i) \overset{(+,+)}{-} (d_{j},t^c_0) \overset{(+,-)}{-} (s^c_1,t^c_1) \overset{(-,+)}{-} (s^c_0,d_i).\]
    If $s^c_1$ is matched to a neighbor worse than $t^c_0$ then the former is an alternating path with two $(+,+)$ edges: these are $(d_i,d_{j})$ and $(t^c_0,s^c_1)$.
    Else $s^c_1$ is matched to $t^c_1$ and the latter is an alternating cycle with a $(+,+)$ edge, which is $(d_i,d_{j})$.
  \item[(2)] $t^c_0$ is matched to $s^c_i$ for some $i \in \{1,2,3\}$: this means $t^c_i$ is matched to a neighbor worse than $s^c_0$ and so $(s^c_0,t^c_i)$ is a
    $(+,+)$ edge and thus we have the following alternating path with two $(+,+)$ edges $(t^c_i,s^c_0)$ and $(d_i,d_j)$:
    \[ (\ast,t^c_i) \overset{(+,+)}{-}  (s^c_0,d_i) \overset{(+,+)}{-} (d_j,\ast).\]
  \item[(3)] $t^c_0$ is matched to either $p^c_4$ or $p^c_7$: we will show that this results in an alternating path with two $(+,+)$ edges.
    Assume without loss of generality that $t^c_0$ is matched to $p^c_4$. 
    Consider the following alternating path: 
    \[ (\ast,s^c_3) \overset{(+,+)}{-} (t^c_0,p^c_4) \overset{(+,+)}{-} (q^c_5,\ast) \ \ \ \ \ \ \ \mathrm{or}\ \ \ \ \ \ \ (\ast,d_j) \overset{(+,+)}{-} (d_i,s^c_0) \overset{(+,-)}{-} (t^c_3,s^c_3) \overset{(-,+)}{-} (t^c_0,p^c_4) \overset{(+,+)}{-} (q^c_5,\ast).\]
    Recall that $s^c_3$ is the top choice neighbor of $t^c_0$ and the vertex $p^c_4$ is the top choice neighbor of $q^c_5$. If the vertex $s^c_3$ is matched to a neighbor worse than $t^c_0$ then the former path is an alternating path in $G_M$ with two $(+,+)$ edges in it: these are
    $(s^c_3,t^c_0)$ and $(p^c_4,q^c_5)$.
    Else $(s^c_3,t^c_3) \in M$ and recall that the edge $(s^c_0,t^c_3)$ is a $(+,-)$ edge as $s^c_0$ prefers $t^c_3$ to $d_i$. This creates the latter
    path which is an alternating path in $G_M$ with 2 $(+,+)$ edges in it: these are $(d_i,d_j)$ and $(p^c_4,q^c_5)$. \qed
  \end{itemize}  
\end{proof}

The gadget $D$ admits 2 popular matchings: $\{(d_0,d_1),(d_2,d_3)\}$ and  $\{(d_0,d_2),(d_1,d_3)\}$. So if $M$ is a popular matching then either
$\{(d_0,d_1),(d_2,d_3)\} \subset M$ or $\{(d_0,d_2),(d_1,d_3)\} \subset M$.
\new{The following lemma further restricts the set of popular edges and this will be used repeatedly in our proof.}

\begin{lemma}
  \label{lem:better-than-d}
  Let $(u,v)$ be an edge in $G$ where both $u$ and $v$ prefer $d_0$ to each other. Then $(u,v)$ cannot be a popular edge.
\end{lemma}  

\begin{proof}
  Let $M$ be a popular matching in $G$ that contains such an edge $(u,v)$. We know from Lemma~\ref{second-lemma} that either $\{(d_0,d_1),(d_2,d_3)\} \subset M$
  or $\{(d_0,d_2),(d_1,d_3)\} \subset M$. So there is always a blocking edge $(d_i,d_j) \in \{(d_1,d_3),(d_1,d_2)\}$ to $M$.
  
  Observe that both $u$ and $v$ cannot belong to the $D$-gadget as there is no such pair within $D$.
  If exactly one of $u,v$ belongs to the $D$-gadget then $(u,v)$ is not a popular edge (by Lemma~\ref{second-lemma}). 
  So neither $u$ nor $v$ belongs to the $D$-gadget and this implies that $u$ prefers $d_0,d_1,d_2,d_3$ to $v$
  and symmetrically, $v$ prefers $d_0,d_1,d_2,d_3$ to $u$.

  Consider the following alternating cycle $C$ with respect to~$M$:
  \[ (u,v) \overset{(+,-)}{-} (d_{i'},d_i) \overset{(+,+)}{-} (d_j,d_{j'}) \overset{(-,+)}{-} (u,v),\]
  where $(d_{i'},d_i)$ and $(d_j,d_{j'})$ are edges from the $D$-gadget in $M$ and $(d_i,d_j)$ is a blocking edge.
  Thus $C$ is an alternating cycle in $G_M$ with a $(+,+)$ edge. This contradicts the popularity of~$M$ (by Theorem~\ref{thm:char-popular}). \qed 
\end{proof}

\begin{corollary}
  \label{cor1}
  The edges $(s^c_0,t^c_0)$ and $(s'^c_0,t'^c_0)$ are not popular edges for any clause $c$.
\end{corollary}  

Corollary~\ref{cor1} follows from Lemma~\ref{lem:better-than-d} by setting $u$ and $v$ to $s^c_0$ and $t^c_0$ (similarly, $s'^c_0$ and $t'^c_0$),
respectively. Let us call $u$ a level~$i$ vertex if $u$ belongs to a level~$i$ gadget. Lemma~\ref{lem:no-pop-general} further restricts the set of popular edges; the proof of this lemma consists of three main
claims.

\begin{lemma}
  \label{lem:no-pop-general}
    No edge between a level~$i$ vertex and a level~$i+1$ vertex is popular, for $0 \le i \le 2$.
\end{lemma}

The proof of Lemma~\ref{lem:no-pop-general} follows from Claims~\ref{lem:no-pop-0-1}-\ref{lem:no-pop-2-3} proved below.
\begin{new-claim}
  \label{lem:no-pop-0-1}
  There is no popular edge between a level~0 vertex and a level~1 vertex.
\end{new-claim}
\begin{proof}
  Let $M$ be a popular matching in $G$ with such an edge, say $(a^c_1,y'_j)$, where $c = X_i \vee X_j \vee X_k$.
  We claim this would create an alternating path in $G_M$
  with two $(+,+)$ edges in it (this would contradict Theorem~\ref{thm:char-popular}). Consider the vertex $b^c_1$.
  There are 3 possibilities for $b^c_1$'s partner in $M$.
  \begin{itemize}
  \item[(1)] $b^c_1$ is matched to $a^c_2$\\
   So $(a^c_2,b^c_2)$ is labeled $(+,+)$. Recall that $b^c_2$ is $a^c_2$'s top choice and the only neighbor that $b^c_2$ prefers to $a^c_2$ is $a^c_1$ (which is matched to $y'_j$). Consider the following alternating path in $G_M$:
    \[ \hspace*{-0.2cm}(\ast,b^c_2) \overset{(+,+)}{-} (a^c_2,b^c_1) \overset{(-,+)}{-} (a^c_1,y'_j) \overset{(+,+)}{-} (x'_j,\ast).\]
    If $x'_j$ is matched to a neighbor worse than $y'_j$ then the above is an alternating path in $G_M$ with two $(+,+)$ edges: these are $(b^c_2,a^c_2)$ and $(y'_j,x'_j)$.
  Else replace $(x'_j,\ast)$ in the above path with $(x'_j,y_j) -  (x_j,\ast)$: the $(+,+)$ edges here are $(b^c_2,a^c_2)$ and $(y_j,x_j)$.

  \smallskip

\item[(2)] $b^c_1$ is matched to $x'_k$\\
  Either the edge $(x'_k,y'_k)$ or the edge $(x_k,y_k)$ will block~$M$. Suppose $y'_k$ is matched to a neighbor worse than $x'_k$ in $M$. Consider the following alternating path in $G_M$:
    \[ \hspace*{-0.2cm}(\ast,y'_k) \overset{(+,+)}{-} (x'_k,b^c_1) \overset{(-,+)}{-} (a^c_1,y'_j) \overset{(+,+)}{-} (x'_j,\ast).\]
    Either the above is an alternating path in $G_M$ with two $(+,+)$ edges or by replacing $(x'_j,\ast)$ with $(x'_j,y_j) - (x_j,\ast)$
    (as done in case~(1)), we get an alternating path in $G_M$ with two $(+,+)$ edges.

    If $y'_k$ is matched to a neighbor better than $x'_k$ in $M$, i.e., if $(x_k,y'_k) \in M$ then prefix
    both these alternating paths with $(\ast,y_k)$.
    This will yield an alternating path in $G_M$ with $(x_k,y_k)$ as a blocking edge and either $(x'_j,y'_j)$ or $(x_j,y_j)$ as a blocking edge.

    \smallskip

    \item[(3)] $b^c_1$ is matched to a neighbor worse than $a^c_1$\\
   Then the edge $(a^c_1,b^c_1)$ is labeled $(+,+)$. Consider the following alternating path in $G_M$:
  \[ (\ast,b^c_1) \overset{(+,+)}{-} (a^c_1,y'_j) \overset{(+,+)}{-} (x'_j,\ast) \ \ \ \ \ \ \ \mathrm{or} \ \ \ \ \ \ \ (\ast,b^c_1) \overset{(+,+)}{-} (a^c_1,y'_j) \overset{(+,-)}{-} (x'_j,y_j) \overset{(+,+)}{-} (x_j,\ast).\]
  That is, if $x'_j$ is matched to a neighbor worse than $y'_j$ then consider the first alternating path above: this is an alternating path in
  $G_M$ with both $(b^c_1,a^c_1)$ and $(y'_j,x'_j)$ as $(+,+)$ edges. Else $(x'_j,y_j) \in M$ and the second alternating path is an alternating path in $G_M$ with $(b^c_1,a^c_1)$ and $(y_j,x_j)$ as $(+,+)$ edges. 
\end{itemize} 

Thus $(a^c_1,y'_j)$ cannot belong to any popular matching in $G$. Similarly, $(b^c_1,x'_k)$ also cannot belong to any popular matching.
Suppose $(b^c_1,x'_k) \in M$.
Consider the vertex $a^c_2$: there are 2 possibilities for $a^c_2$'s partner in $M$. 
\begin{enumerate}
\item[(1)] $a^c_2$ is matched to $b^c_2$\\ 
Recall that $a^c_1$ is $b^c_2$'s top choice and since $a^c_1$ is not matched to either $y'_j$ (by our proof above) or $b^c_1$ (which is
matched to $x'_k$), the edge $(a^c_1,b^c_2)$ is labeled $(+,+)$. Consider the following alternating path in $G_M$:
    \[ \hspace*{-0.2cm}(\ast,a^c_1) \overset{(+,+)}{-} (b^c_2,a^c_2) \overset{(-,+)}{-} (b^c_1,x'_k) \overset{(+,+)}{-} (y'_k,\ast).\]
    If $y'_k$ is matched to a neighbor worse than $x'_k$ then the above is an alternating path in $G_M$ with two $(+,+)$ edges: these are 
$(a^c_1,b^c_2)$ and $(x'_k,y'_k)$. Else replace $(y'_k,\ast)$ in the above path with $(y'_k,x_k) -  (y_k,\ast)$: the $(+,+)$ edges here 
are $(a^c_1,b^c_2)$ and $(x_k,y_k)$.

\smallskip

\item[(2)] $a^c_2$ is matched to a neighbor worse than $b^c_1$\\ 
Since $a^c_2$ is $b^c_1$'s top choice, the edge $(a^c_2,b^c_1)$ is labeled $(+,+)$. Consider the following alternating path in $G_M$:
  \[ (\ast,a^c_2) \overset{(+,+)}{-} (b^c_1,x'_k) \overset{(+,+)}{-} (y'_k,\ast) \ \ \ \ \ \ \ \mathrm{or} \ \ \ \ \ \ \ (\ast,a^c_2) \overset{(+,+)}{-} (b^c_1,x'_k) \overset{(+,-)}{-} (y'_k,x_k) \overset{(+,+)}{-} (y_k,\ast).\]
  That is, if $y'_k$ is matched to a neighbor worse than $x'_k$ then consider the first alternating path above: this is an alternating path in
  $G_M$ with both $(a^c_2,b^c_1)$ and $(x'_k,y'_k)$ as $(+,+)$ edges. Else $(y'_k,x_k) \in M$ and the second alternating path is an alternating path in $G_M$ with $(a^c_2,b^c_1)$ and $(x_k,y_k)$ as $(+,+)$ edges. 
\end{enumerate}
Thus $(b^c_1,x'_k)$ is also not a popular edge. The proof that $(a^c_3,y'_k),(a^c_5,y'_i),(a'^c_1,y'_k),(a'^c_3,y'_i),(a'^c_5,y'_j)$ 
are not popular edges is analogous to the proof that $(a^c_1,y'_j)$ is not a popular edge. Similarly, the proof that $(b^c_3,x'_i),(b^c_5,x'_j),(b'^c_1,x'_j),(b'^c_3,x'_k),(b'^c_5,x'_i)$ are not popular edges is analogous to the proof that $(b^c_1,x'_k)$ is not a popular edge.
The presence of any other edge between a level~0 vertex corresponding to clause $c$ and a level~1 vertex in the popular matching $M$ would contradict Lemma~\ref{lem:better-than-d}. \qed
\end{proof}  

\begin{new-claim}
  \label{lem:no-pop-1-2}
  There is no popular edge between a level~1 vertex and a level~2 vertex.
\end{new-claim}
\begin{proof}
  Let $M$ be a popular matching in $G$ that contains such an edge, say $(p^c_2,y_j)$, where $c = X_i \vee X_j \vee X_k$. Consider the following alternating path with respect to~$M$:
  \[ (p^c_2,y_j) \overset{(+,+)}{-} (x'_j,y'_j) \overset{(+,+)}{-} (x_j,\ast).\]
  Since $M$ is a perfect matching, $x'_j$ is matched in $M$. 
  We know that no edge between $x'_j$ and a level~0 vertex belongs to $M$ (by Claim~\ref{lem:no-pop-0-1}). Also, $M$ cannot match $x'_j$ to a neighbor that
  it regards worse than $d_0$ (by Lemma~\ref{lem:better-than-d}) or to any neighbor in the $D$-gadget (by Lemma~\ref{second-lemma}). Thus $x'_j$ has to be matched to $y'_j$ in $M$ and so the above alternating path has two $(+,+)$
  edges: $(x'_j,y_j)$ and $(x_j,y'_j)$. This is a contradiction to $M$'s popularity (by Theorem~\ref{thm:char-popular}). 

It can similarly be shown that $(q^c_2,x_k)$ cannot belong to $M$. We consider the alternating path
$(q^c_2,x_k) - (y'_k,x'_k) - (y_k,\ast)$ here to get a contradiction to Theorem~\ref{thm:char-popular}.

The proof that $M$ cannot include any other edge between  a level~1 vertex and a level~2 vertex corresponding to clause~$c$ is
either analogous to the above proof or the presence of such an edge in $M$ violates Lemma~\ref{lem:better-than-d}. \qed
\end{proof}

\begin{new-claim}
  \label{lem:no-pop-2-3}
  There is no popular edge between a level~2 vertex and a level~3 vertex.
\end{new-claim}
\begin{proof}
  Let $M$ be a popular matching in $G$ that contains such an edge, say $(s^c_0,q^c_0)$, for some clause~$c$.
  Consider the following alternating path with respect to~$M$:
  \[ (s^c_0,q^c_0) \overset{(+,+)}{-} (p^c_2,q^c_2) \overset{(+,+)}{-} (p^c_0,\ast) \ \ \ \ \ \ \ \ \mathrm{or}\ \ \ \ \ \ \ \ \ (s^c_0,q^c_0) \overset{(+,+)}{-} (p^c_2,q^c_1) \overset{(+,+)}{-} (p^c_1,\ast).\]
  The vertex $p^c_2$ is matched in~$M$ and
  its partner cannot be a level~1 vertex (by Claim~\ref{lem:no-pop-1-2}) or a neighbor worse than $d_0$ (by Lemma~\ref{lem:better-than-d}) or to any neighbor in the $D$-gadget (by Lemma~\ref{second-lemma}).
  So either $(p^c_2,q^c_2)$ or $(p^c_2,q^c_1)$ is in $M$.
  This means either the first alternating path given above or the second one is an alternating path in $G_M$ with two $(+,+)$ edges: $(p^c_2,q^c_0)$ and
  $(p^c_0,q^c_2)$ in the former and $(p^c_2,q^c_0)$ and $(p^c_1,q^c_1)$ in the latter. This is a contradiction to $M$'s popularity (by Theorem~\ref{thm:char-popular}).

The proof that $M$ cannot include any other edge between a level~2 vertex and a level~3 vertex is either analogous to the
  above proof or the presence of such an edge in $M$ violates Lemma~\ref{lem:better-than-d}.  \qed
\end{proof}

\new{We now show that no popular matching contains an inter-gadget edge incident to any vertex in the $Z$-gadget. This is analogous to Lemma~\ref{second-lemma}, which showed the same property for the $D$-gadget.}

\begin{lemma}
  \label{lem:z_outside}
  All popular matchings match the 6 vertices of the $Z$-gadget among themselves.	
\end{lemma}
\begin{proof}
  Let $M$ be any popular matching in $G$. It follows from Lemma~\ref{lem:better-than-d} that $M$ has to match $z_2,z_3, z_4$, and $z_5$ within
  the $Z$-gadget. Let us now show that $z_0$ also has to be matched within the $Z$-gadget. Then it immediately follows that $z_1$ also has to be
  matched within the $Z$-gadget. We have the following 3 cases: 
	\begin{itemize}
		\item[(1)] Suppose $z_0$ is matched in $M$ to a level~0 neighbor, say $b^c_1$, for some clause~$c$. Then $(a^c_1,b^c_1)$ is a blocking edge to~$M$. Lemmas~\ref{second-lemma}, \ref{lem:better-than-d}, and \ref{lem:no-pop-general} ensure that $a^c_1$ is either matched to $z_1$ or to $b^c_2$. We investigate these two cases below.
			\begin{itemize}
	          \item $(a^c_1,z_1) \in M$: Here both $z_0$ and $z_1$ are matched to vertices they prefer to all their neighbors inside the $Z$-gadget, except for $z_4$ and~$z_5$. We know that $z_4$ and $z_5$ must be matched inside the $Z$-gadget. There are 3 subcases and in each case there is an alternating cycle in $G_M$ with a blocking edge $(a^c_1,b^c_1)$: a contradiction to $M$'s popularity (by Theorem~\ref{thm:char-popular}).               
                  \begin{itemize}
                  \item $(z_4,z_2) \in M$: the alternating cycle is $(b^c_1,z_0) \overset{(+,-)}{-} (z_4, z_2) \overset{(+,-)}{-} (z_1,a^c_1) \overset{(+,+)}{-} (b^c_1,z_0)$.
                  \item $(z_4,z_3) \in M$: the alternating cycle is $(b^c_1,z_0) \overset{(+,-)}{-} (z_4, z_3) \overset{(+,-)}{-} (z_1,a^c_1) \overset{(+,+)}{-} (b^c_1,z_0)$.
                  \item $(z_4,z_5) \in M$: the alternating cycle is $(b^c_1,z_0) \overset{(+,-)}{-} (z_4, z_5) \overset{(-,+)}{-} (z_1,a^c_1) \overset{(+,+)}{-} (b^c_1,z_0)$.
                  \end{itemize}
			
                       \medskip						
	              
                       \item $(a^c_1,b^c_2) \in M$: Lemmas~\ref{second-lemma}, \ref{lem:better-than-d}, and \ref{lem:no-pop-general}
                        ensure that $a^c_2$ is matched to~$z_1$ (recall that $M$ is perfect). 
                        This leads to the same 3 subcases as above, except that instead of the
                        edge $(z_1,a^c_1)$, there is the path $(z_1, a^c_2) - (b^c_2,a^c_1)$ in $G_M$: here $(a^c_2,b^c_2)$ is labeled $(+,-)$.
                        \end{itemize}

                        \smallskip
                        
	        \item[(2)] Suppose $z_0$ is matched in $M$ to a level~1 neighbor, say $y_i$, for some $i \in \{1,\ldots,\kappa\}$.
                  
		  This case is similar to the previous case. Here the edge $(x_i, y_i)$ becomes the blocking edge to $M$. It follows from Lemmas~\ref{second-lemma}, \ref{lem:better-than-d}, and \ref{lem:no-pop-general} that $x_i$ is either matched to $z_1$ or to $y'_i$. The latter case leaves $x'_i$ unmatched and the subcases that arise in the former case are
                  analogous to the ones in case~(1).
                  
		\item[(3)] Suppose $z_0$ is matched in $M$ to a level~2 neighbor, say $q^c_0$, for some clause~$c$.

                  It follows from Lemmas~\ref{second-lemma}, \ref{lem:better-than-d}, and \ref{lem:no-pop-general} that $(p^c_0,q^c_2), (p^c_2,q^c_1)$, and $(p^c_1,z_1)$ are in $M$.
                  Consider the alternating path 
                  $(z_0,q^c_0) - (p^c_2, q^c_1) - (p^c_1,z_1)$: it has two blocking edges
                  $(p^c_2,q^c_0)$ and $(p^c_1,q^c_1)$. This is again a contradiction to $M$'s popularity.
		\end{itemize}	
		
        Recall that  Lemma~\ref{second-lemma} showed that all vertices of the $D$-gadget must be matched within the gadget. Thus $z_0$ cannot be matched to a vertex in the
        $D$-gadget. The case where  $z_0$ is matched in $M$ to a level~3 neighbor does not arise as such an edge would violate Lemma~\ref{lem:better-than-d}.
        This finishes our proof that any popular matching $M$ matches the 6 vertices of the $Z$-gadget among themselves. \qed
\end{proof}

\new{Our next lemma shows that the $Z$-gadget has a unique popular matching. The fact that the edge $(z_0,z_1)$ has to be in any popular matching $M$ will be used repeatedly in Section~\ref{sec:states}.}

\begin{lemma}
	The only popular matching inside the $Z$-gadget is $\left\{ (z_0,z_1), (z_2,z_3), (z_4,z_5) \right\}$.
	\label{lem:z_inside}
\end{lemma}

\begin{proof}
  The matching $\{(z_0,z_1), (z_2,z_3), (z_4,z_5)\}$ is stable in the $Z$-gadget, thus this is a popular matching. Note that this gadget has no other stable matching.

  Let $M$ be any popular matching. We know from Lemma~\ref{lem:z_outside} that $M$ matches the 6 vertices of the $Z$-gadget among themselves. Suppose $M$ contains one or more of the edges $(z_i,z_j)$
  where $i = j \bmod 2$  (colored black in Fig.~\ref{fig:z}).
This means one of the edges $(z_0,z_2), (z_2,z_4),(z_0,z_4)$ is in $M$.
Let $(z_0,z_2) \in M$. There are three candidate matchings that we need to check for popularity: note
that none is popular (by Theorem~\ref{thm:char-popular}).
	\begin{itemize}
		\item $\left\{ (z_0,z_2), (z_1,z_3), (z_4,z_5) \right\}$: this has the alternating cycle $(z_2,z_0) \overset{(+,+)}{-} (z_1,z_3) \overset{(-,+)}{-} (z_5,z_4) \overset{(+,-)}{-} (z_2,z_0)$ with the blocking edge $(z_0,z_1)$.
		\item $\left\{ (z_0,z_2), (z_1,z_4), (z_3,z_5) \right\}$: this has the alternating cycle $(z_0, z_2) \overset{(-,+)}{-} (z_3, z_5) \overset{(-,+)}{-} (z_1, z_4) \overset{(+,+)}{-} (z_0, z_2)$ with the blocking edge $(z_0,z_4)$.
		\item $\left\{ (z_0,z_2), (z_1,z_5), (z_3,z_4) \right\}$: this has the alternating cycle $(z_2,z_0) \overset{(+,-)}{-} (z_1,z_5) \overset{(+,+)}{-} (z_3, z_4) \overset{(+,-)}{-} (z_2, z_0)$ with the blocking edge $(z_3,z_5)$.
	\end{itemize}

        So $(z_0,z_2)$ does not belong to $M$.
        We now claim that neither $(z_2,z_4)$ nor $(z_4,z_0)$ also belongs to~$M$.
        This is because when confined to edges within the $Z$-gadget, the $Z$-gadget is symmetric with respect to $z_0,z_2,z_4$
        (similarly, wrt $z_1,z_3,z_5$). Thus the same analysis as shown above for $(z_0,z_2)$ holds for $(z_2,z_4)$ and $(z_4,z_0)$ as well
        by replacing every subscript $s$ with $(s+2)\bmod 6$ for the former edge and with $(s+4)\bmod 6$ for the latter edge.
        Hence we can conclude that
        $M \subset \left\{ z_0, z_2, z_4 \right\}\times\left\{ z_1, z_3, z_5 \right\}$.

Suppose $M$ contains an unstable edge here (dotted and gray in Fig.~\ref{fig:z}), say $(z_0,z_3)$: among the vertices in the
        $Z$-gadget, $z_3$ is the last choice of $z_0$ and the edge $(z_0,z_2)$ blocks $M$. Since $z_2$ has to be matched in $M$, there are two cases.
\begin{itemize}
\item $(z_1, z_2) \in M$: the 4 vertices $z_0,z_1,z_2,z_3$ prefer $\{(z_0,z_2),(z_1,z_3)\}$ to $\{(z_0,z_3),(z_1,z_2)\} \subset M$.
\item $(z_2, z_5 )\in M$: the 4 vertices $z_0,z_2,z_3,z_5$ prefer $\{(z_0,z_2),(z_3,z_5)\}$ to $\{(z_0,z_3),(z_2,z_5)\} \subset M$.
\end{itemize}
Thus in both cases we have a contradiction to $M$'s popularity.
Analogous proofs hold for other  unstable edges  chosen from $\{z_0,z_2,z_4\}\times\{z_1,z_3,z_5\}$. Thus the only popular matching inside the $Z$-gadget is
$\left\{ (z_0,z_1), (z_2,z_3), (z_4,z_5) \right\}$. \qed
\end{proof}

\section{Stable states versus unstable states}
\label{sec:states}
In this section we will show how to obtain a 1-in-3 satisfying truth assignment for the input $B$ from any popular matching in $G$.
The following definition will be useful to us.

\begin{definition}
A gadget $A$ in $G = (V,E)$ is said to be in {\em unstable state} with respect to matching $M$ if there is a blocking edge
$(u,v) \in V(A) \times V(A)$ with respect to $M$. If there is no such blocking edge to $M$ then we say $A$ is in {\em stable state} with respect to $M$.
\end{definition}

In Figures~\ref{fig:1}-\ref{fig:z} depicting our gadgets, corresponding to matchings that consist of colored edges within the gadget,
the relevant gadget is in {\em stable} state. A level~1 gadget in unstable state will encode the 
corresponding variable being set to true while a level~1 gadget in stable state will encode the corresponding variable being set to false.
We will now analyze what gadgets are in stable/unstable state with respect to any popular matching  $M$ in $G$. This will lead to the proof
that for any clause $c$, exactly one of the level~1 gadgets corresponding to the three variables in $c$ is in unstable state. 

Lemma~\ref{lemma1} takes the starting step in this proof. Here $M$ is any popular matching in $G$.

\begin{lemma}
\label{lemma1}
For any clause $c$, the following statements hold:
\begin{itemize}
   \item all its 6 level~0 gadgets are in stable state with respect to $M$;
   \item both its level~3 gadgets in $G$ are in unstable state with respect to $M$.
\end{itemize}  
\end{lemma}
\begin{proof}
  Consider a level~0 gadget corresponding to clause $c$, say the one on vertices $a^c_1,b^c_1,a^c_2,b^c_2$.
  Lemmas~\ref{second-lemma}, \ref{lem:better-than-d}, \ref{lem:no-pop-general}, and \ref{lem:z_outside} imply
  that either $\{(a^c_1,b^c_1),(a^c_2,b^c_2)\} \subset M$ or 
  $\{(a^c_1,b^c_2),(a^c_2,b^c_1)\} \subset M$. Thus there is no blocking edge within this gadget. 
  As this holds for every level~0 gadget corresponding to $c$ and for every clause $c$, the first part of the lemma follows. 
  
  We will now prove the second part of the lemma.
  Since $M$ is a perfect matching, the vertices $s^c_0,t^c_0$ (also $s'^c_0,t'^c_0$) have to be matched in $M$, for all clauses $c$. It follows from
  Lemmas~\ref{second-lemma} and \ref{lem:better-than-d} that both $s^c_0$ and $t^c_0$ (similarly, $s'^c_0$ and $t'^c_0$) have to be matched to neighbors that are better than $d_0$.
  Lemma~\ref{lem:no-pop-general} showed that there is no popular edge between a level~3 vertex and a level~2 vertex. Thus
  $s^c_0$ is matched to $t^c_i$ for some $i \in \{1,2,3\}$.

  If $s^c_0$ is matched to $t^c_i$ then $s^c_i$ has to be matched to $t^c_0$---otherwise
  Lemma~\ref{lem:better-than-d} would be violated by $s^c_i$ and its partner.
  So $(s^c_i,t^c_i)$ blocks $M$ and this holds for every clause~$c$.
  Similarly, there is a blocking edge $(s'^c_j,t'^c_j)$ for some $j \in \{1,2,3\}$ for every clause~$c$. \qed
\end{proof}

Lemma~\ref{lemma2} is our main technical lemma. This will be used in Lemma~\ref{lemma3} to show that at least one of the level~1 gadgets corresponding to the three variables in clause~$c$ is in unstable state.

\begin{lemma}
\label{lemma2}
For any clause $c$, at least {\em one} of the following two conditions has to hold:
\begin{itemize} 
\item two or more of its first three level~2 gadgets are in unstable state with respect to~$M$;
\item two or more of its last three level~2 gadgets are in unstable state with respect to~$M$.
\end{itemize}
\end{lemma}

\begin{proof}
Suppose both statements are false. Let $M$ be a popular matching and $c$ be a clause such that corresponding to clause $c$, at least two among its {\em first}
three level~2 gadgets are in stable state with respect to $M$ {\em and} at least two among its {\em last} three level~2 gadgets are in stable state with respect 
to $M$.

Consider the two level~3 gadgets corresponding to $c$. We know that $(s^c_0,t^c_i), (s^c_i,t^c_0)$ are in $M$ for some $i \in \{1,2,3\}$ and
similarly, $(s'^c_0,t'^c_j), (s'^c_j,t'^c_0)$ are in $M$ for some $j \in \{1,2,3\}$ (see the proof of Lemma~\ref{lemma1}). We will now show the existence of an alternating path $\rho$ that will contradict $M$'s popularity.

For this, we claim it suffices to show that one gadget in each of the following two sets of gadgets is in {\em stable state}:
\begin{itemize}
\item the {\em first} three level~2 gadgets with a vertex that either $s^c_0$ or $t^c_0$ prefers to its partner in $M$;
\item the {\em last} three level~2 gadgets with a vertex that either $s'^c_0$ or $t'^c_0$ prefers to its partner in $M$.
\end{itemize}

For instance, suppose $i = 1$ and $j = 2$.
So $t^c_0$ prefers $p^c_4$ and $p^c_7$ to its partner $s^c_1$ in $M$ and $s'^c_0$ prefers $q'^c_0$ to its partner $t'^c_2$ in $M$ 
and $t'^c_0$ prefers $p'^c_7$ to its partner $s'^c_2$ in $M$. Consider the level~2 gadgets containing $p^c_4, p^c_7, q'^c_0$, and $p'^c_7$. 

Observe that
by our assumption in the first paragraph above, either the gadget of $p^c_4$ or the gadget of $p^c_7$ is in stable state, similarly either the gadget of $q'^c_0$
or the gadget of $p'^c_7$ is in stable state. In all 4 cases, we will show the existence of an
alternating path $\rho$ in $G_M$ with {\em two} blocking edges $(s^c_1,t^c_1)$ and $(s'^c_2,t'^c_2)$,
which is a contradiction to $M$'s popularity (by Theorem~\ref{thm:char-popular}).

\begin{enumerate}
\item Suppose the gadgets of $p^c_4$ and $q'^c_0$ are in stable state. So the edges $(p^c_{\ell},q^c_{\ell}) \in M$ for $\ell=3,4,5$
and the edges $(p'^c_h,q'^c_h) \in M$ for $h=0,1,2$. Consider the following
alternating path $\rho$ with respect to~$M$: 
\begin{gather*}
(s^c_0,t^c_1) \overset{(+,+)}{-}
(s^c_1,t^c_0) \overset{(+,-)}{-}
(p^c_4,q^c_4) \overset{(-,+)}{-} 
(p^c_5,q^c_5) \overset{(+,-)}{-} 
(p^c_3,q^c_3) \overset{(-,+)}{-}\\
(z_0,z_1) \overset{(+,-)}{-} 
(p'^c_1,q'^c_1) \overset{(-,+)}{-} 
(p'^c_2,q'^c_2) \overset{(+,-)}{-} 
(p'^c_0,q'^c_0) \overset{(-,+)}{-} 
(s'^c_0,t'^c_2) \overset{(+,+)}{-}
(s'^c_2,t'^c_0).
\end{gather*}

We know that $(z_0,z_1) \in M$ (by Lemma~\ref{lem:z_inside}).
Note that $\rho$ is the desired alternating path in $G_M$ with two blocking edges $(s^c_1,t^c_1)$ and $(s'^c_2,t'^c_2)$.

\item Suppose the gadgets of $p^c_4$ and $p'^c_7$ are in stable state. So the edges $(p^c_{\ell},q^c_{\ell}) \in M$ for $\ell=3,4,5$
and the edges $(p'^c_h,q'^c_h) \in M$ for $h=6,7,8$. Consider the following
alternating path $\rho$ with respect to~$M$: 
\begin{gather*}
(s^c_0,t^c_1) \overset{(+,+)}{-}
(s^c_1,t^c_0) \overset{(+,-)}{-}
(p^c_4,q^c_4) \overset{(-,+)}{-} 
(p^c_5,q^c_5) \overset{(+,-)}{-} 
(p^c_3,q^c_3) \overset{(-,+)}{-}\\
(z_0,z_1) \overset{(+,-)}{-} 
(q'^c_6,p'^c_6) \overset{(-,+)}{-} 
(q'^c_8,p'^c_8) \overset{(+,-)}{-} 
(q'^c_7,p'^c_7) \overset{(-,+)}{-} 
(t'^c_0,s'^c_2) \overset{(+,+)}{-}
(t'^c_2,s'^c_0).
\end{gather*}

Observe that the labels on edges of $\rho\setminus M$ are identical to the first case and
thus $\rho$ is the desired alternating path in $G_M$ with two blocking edges $(s^c_1,t^c_1)$ and 
$(s'^c_2,t'^c_2)$.

\item Suppose the gadgets of $p^c_7$ and $q'^c_0$ are in stable state. So the edges $(p^c_{\ell},q^c_{\ell}) \in M$ for $\ell=6,7,8$ and the edges $(p'^c_h,q'^c_h) \in M$ for $h=0,1,2$. Consider the following
alternating path $\rho$ with respect to~$M$: 
\begin{gather*}
(s^c_0,t^c_1) \overset{(+,+)}{-}
(s^c_1,t^c_0) \overset{(+,-)}{-}
(p^c_7,q^c_7) \overset{(-,+)}{-} 
(p^c_8,q^c_8) \overset{(+,-)}{-} 
(p^c_6,q^c_6) \overset{(-,+)}{-}\\
(z_0,z_1) \overset{(+,-)}{-} 
(p'^c_1,q'^c_1) \overset{(-,+)}{-} 
(p'^c_2,q'^c_2) \overset{(+,-)}{-} 
(p'^c_0,q'^c_0) \overset{(-,+)}{-} 
(s'^c_0,t'^c_2) \overset{(+,+)}{-}
(s'^c_2,t'^c_0).
\end{gather*}

Again, observe that the labels on edges of $\rho\setminus M$ are identical to the first two cases and
$\rho$ is the desired alternating path with two blocking edges $(s^c_1,t^c_1)$ and $(s'^c_2,t'^c_2)$.

\item Suppose the gadgets of $p^c_7$ and $p'^c_7$ are in stable state. So the edges $(p^c_{\ell},q^c_{\ell}) \in M$ for $\ell=6,7,8$ and the edges $(p'^c_h,q'^c_h) \in M$ for $h=6,7,8$. Consider the following
alternating path $\rho$ with respect to~$M$: 
\begin{gather*}
(s^c_0,t^c_1) \overset{(+,+)}{-}
(s^c_1,t^c_0) \overset{(+,-)}{-}
(p^c_7,q^c_7) \overset{(-,+)}{-} 
(p^c_8,q^c_8) \overset{(+,-)}{-} 
(p^c_6,q^c_6) \overset{(-,+)}{-}\\
(z_0,z_1) \overset{(+,-)}{-} 
(q'^c_6,p'^c_6) \overset{(-,+)}{-} 
(q'^c_8,p'^c_8) \overset{(+,-)}{-} 
(q'^c_7,p'^c_7) \overset{(-,+)}{-} 
(t'^c_0,s'^c_2) \overset{(+,+)}{-}
(t'^c_2,s'^c_0).
\end{gather*}

\end{enumerate}

As before, the labels on edges of $\rho\setminus M$ are identical to the above three 
cases and $\rho$ is the desired alternating path with two blocking edges $(s^c_1,t^c_1)$ and $(s'^c_2,t'^c_2)$.

\smallskip

For any $(i,j) \in \{1,2,3\}\times\{1,2,3\}$, an analogous construction can be shown.
\begin{itemize}
\item Let $i = j = 1$. So $t^c_0$ prefers $p^c_4$ and $p^c_7$ to its partner $s^c_1$ in $M$ and 
$t'^c_0$ prefers $p'^c_4$ and  $p'^c_7$ to its partner $s'^c_1$ in $M$. 
We know that either the gadget of $p^c_4$ or the gadget of $p^c_7$ is in stable state, and 
similarly, either the gadget of $p'^c_4$ or the gadget of $p'^c_7$ is in stable state.

Suppose the gadgets of $p^c_4$ and $p'^c_4$ are in stable state. So the edges $(p^c_{\ell},q^c_{\ell}) \in M$ 
for $\ell=3,4,5$ and the edges $(p'^c_h,q'^c_h) \in M$ for $h=3,4,5$. Consider the following
alternating path $\rho$ with respect to~$M$: 
\begin{gather*}
(s^c_0,t^c_1) \overset{(+,+)}{-}
(s^c_1,t^c_0) \overset{(+,-)}{-}
(p^c_4,q^c_4) \overset{(-,+)}{-} 
(p^c_5,q^c_5) \overset{(+,-)}{-} 
(p^c_3,q^c_3) \overset{(-,+)}{-}\\
(z_0,z_1) \overset{(+,-)}{-} 
(q'^c_3,p'^c_3) \overset{(-,+)}{-} 
(q'^c_5,p'^c_5) \overset{(+,-)}{-} 
(q'^c_4,p'^c_4) \overset{(-,+)}{-} 
(t'^c_0,s'^c_1) \overset{(+,+)}{-}
(t'^c_1,s'^c_0).
\end{gather*}

Observe again that the labels on edges of $\rho\setminus M$ are identical to the labels
obtained for the desired alternating paths when $i = 1$ and $j = 2$. The path $\rho$ is the desired 
alternating path in $G_M$ with two blocking edges $(s^c_1,t^c_1)$ and $(s'^c_1,t'^c_1)$.

The case when the gadgets of $p^c_7$ and $p'^c_7$ are in stable state was already seen in Case~4 
of  $i = 1$ and $j = 2$. The only difference between the path that we will construct now with the path $\rho$ 
seen there is in the last two edges: now we will have $(t'^c_0,s'^c_1)$ and $(t'^c_1,s'^c_0)$ in $M$; thus the 
blocking edges to our path will be  $(s^c_1,t^c_1)$ and $(s'^c_1,t'^c_1)$.

The proofs for the remaining two cases: (i)~when the gadgets of $p^c_4$ and 
$p'^c_7$ are in stable state and (ii)~when the gadgets of $p^c_7$ and 
$p'^c_4$ are in stable state are totally analogous to the above case. 
Thus in all 4 cases, we can show the existence of an alternating path $\rho$ in $G_M$ with {\em two}
blocking edges $(s^c_1,t^c_1)$ and $(s'^c_1,t'^c_1)$: a contradiction to $M$'s popularity. 

\item Let $i = 1$ and $j = 3$. So $t^c_0$ prefers $p^c_4$ and $p^c_7$ to its partner $s^c_1$ in $M$ 
and $s'^c_0$ prefers $q'^c_0$ and  $q'^c_3$ to its partner $t'^c_3$ in $M$. 
We know that either the gadget of $p^c_4$ or the gadget of $p^c_7$ is in stable state, and 
similarly, either the gadget of $q'^c_0$ or the gadget of $q'^c_3$ is in stable state. 

The cases when the gadgets of $p^c_4$ and $q'^c_0$ are in stable state 
and when the gadgets of $p^c_7$ and $q'^c_0$ are in stable state
were already seen in Case~1 and Case~3 of
$i = 1$ and $j = 2$: thus we can construct analogous alternating paths in these cases. 
So let us consider the case when the gadgets of $p^c_7$ and $q'^c_3$ are in stable state.

So the edges $(p^c_{\ell},q^c_{\ell}) \in M$ for $\ell=6,7,8$ and the edges $(p'^c_h,q'^c_h) \in M$ for $h=3,4,5$.
Consider the following alternating path $\rho$ with respect to~$M$: 
\begin{gather*}
(s^c_0,t^c_1) \overset{(+,+)}{-}
(s^c_1,t^c_0) \overset{(+,-)}{-}
(p^c_7,q^c_7) \overset{(-,+)}{-} 
(p^c_8,q^c_8) \overset{(+,-)}{-} 
(p^c_6,q^c_6) \overset{(-,+)}{-}\\
(z_0,z_1) \overset{(+,-)}{-} 
(p'^c_4,q'^c_4) \overset{(-,+)}{-} 
(p'^c_5,q'^c_5) \overset{(+,-)}{-} 
(p'^c_3,q'^c_3) \overset{(-,+)}{-} 
(s'^c_0,t'^c_3) \overset{(+,+)}{-}
(s'^c_3,t'^c_0).
\end{gather*}

The above path $\rho$ is the desired alternating path in $G_M$
with two blocking edges $(s^c_1,t^c_1)$ and $(s'^c_3,t'^c_3)$.

The remaining case, i.e., when the gadgets of $p^c_4$ and $q'^c_3$ are in stable state,
is totally analogous to the above case and we can again show an alternating path in $G_M$ with two blocking edges $(s^c_1,t^c_1)$ and $(s'^c_3,t'^c_3)$.

\item Let $i = 2$ and $j = 1$. This is a ``mirror image'' of the very first case considered:
when $i = 1$ and $j  =2$. The only difference is that we swap {\em primed variables} and
{\em unprimed variables} in $\rho$. For example, when the gadgets of $q^c_0$ and $p'^c_4$ are 
in stable state, the desired alternating path is exactly the same as $\rho$ in Case~1 there, except
for this swapping of roles. Thus, $\rho$ with blocking edges $(s^c_2,t^c_2)$ and $(s'^c_1,t'^c_1)$, would be:
\begin{gather*}
(t^c_0,s^c_2) \overset{(+,+)}{-}
(t^c_2,s^c_0) \overset{(+,-)}{-}
(q^c_0,p^c_0) \overset{(-,+)}{-} 
(q^c_2,p^c_2) \overset{(+,-)}{-} 
(q^c_1,p^c_1) \overset{(-,+)}{-}\\
(z_0,z_1) \overset{(+,-)}{-} 
(q'^c_3,p'^c_3) \overset{(-,+)}{-} 
(q'^c_5,p'^c_5) \overset{(+,-)}{-} 
(q'^c_4,p'^c_4) \overset{(-,+)}{-} 
(t'^c_0,s'^c_1) \overset{(+,+)}{-}
(t'^c_1,s'^c_0).
\end{gather*}

\item Let $i = j = 2$. So $s^c_0$ prefers $q^c_0$ to its partner $t^c_2$ in $M$ and
and $t^c_0$ prefers $p^c_7$ to its partner $s^c_2$ in $M$ and 
$s'^c_0$ prefers $q'^c_0$ to its partner $t'^c_2$ in $M$ and
$t'^c_0$ prefers $p'^c_7$ to its partner $s'^c_1$ in $M$. 
We know that either the gadget of $q^c_0$ or the gadget of $p^c_7$ is in stable state, and 
similarly, either the gadget of $q'^c_0$ or the gadget of $p'^c_7$ is in stable state. 

The cases when the gadgets of $p^c_7$ and $q'^c_0$ are in stable state and when the gadgets of 
$p^c_7$ and $p'^c_7$ are in stable state were already seen in Cases~3 and 4 of $i = 1$ and $j = 2$. 
Let us consider the case when the gadgets of $q^c_0$ and $q'^c_0$ are in 
stable state.

So the edges $(p^c_{\ell},q^c_{\ell}) \in M$ for $\ell=0,1,2$ and the edges $(p'^c_h,q'^c_h) \in M$ for $h=0,1,2$. 
Consider the following alternating path $\rho$ with respect to~$M$: 
\begin{gather*}
(t^c_0,s^c_2) \overset{(+,+)}{-}
(t^c_2,s^c_0) \overset{(+,-)}{-}
(q^c_0,p^c_0) \overset{(-,+)}{-} 
(q^c_2,p^c_2) \overset{(+,-)}{-} 
(q^c_1,p^c_1) \overset{(-,+)}{-}\\
(z_0,z_1) \overset{(+,-)}{-} 
(p'^c_1,q'^c_1) \overset{(-,+)}{-} 
(p'^c_2,q'^c_2) \overset{(+,-)}{-} 
(p'^c_0,q'^c_0) \overset{(-,+)}{-} 
(s'^c_0,t'^c_2) \overset{(+,+)}{-}
(s'^c_2,t'^c_0).
\end{gather*}

The path $\rho$ is the desired 
alternating path in $G_M$ with two blocking edges $(s^c_2,t^c_2)$ and $(s'^c_2,t'^c_2)$.
The proof for the remaining case is totally analogous.
Thus in all 4 cases, we can show the existence of an alternating path $\rho$ in $G_M$ with {\em two}
blocking edges $(s^c_2,t^c_2)$ and $(s'^c_2,t'^c_2)$: a contradiction to $M$'s popularity. 

It is easy to see that the remaining cases of $(i,j)$ are totally analogous to the ones listed above and this finishes the proof of the lemma. \qed
\end{itemize}
\end{proof}

Recall that there are three level~1 gadgets associated with any clause $c$: these gadgets correspond to the three variables in $c$.

\begin{lemma}
\label{lemma3}
  Let $c = X_i \vee X_j \vee X_k$.
At least one of the level~1 gadgets corresponding to $X_i,X_j,X_k$ is in unstable state with respect to $M$.
\end{lemma}
\begin{proof}
Suppose not. That is, assume that for some clause $c$, all three of its level~1 gadgets are in stable state. Let 
$c = X_i \vee X_j \vee X_k$. So $(x_r,y_r)$ and $(x'_r,y'_r)$ are in $M$ for all $r \in \{i,j,k\}$.

We know from 
Lemma~\ref{lemma2} that either two or more of the {\em first} three level~2 gadgets corresponding to $c$ are in unstable state with 
respect to $M$; or two or more of the {\em last} three level~2 gadgets corresponding to $c$ are in unstable state with respect to~$M$. 
We assume without loss of generality that the first and second gadgets, i.e., those on $p^c_i,q^c_i$, for $0 \le i \le 5$,  are in unstable
state with respect to~$M$.

We know from our lemmas in Section~\ref{sec:pop-edges} that there is no popular edge across gadgets. Thus $M$ matches the 6 vertices of a level~2 gadget 
with each other. In particular, it follows from Lemma~\ref{lem:better-than-d} that for the level~2 gadget on $p^c_i,q^c_i$ for $i = 0,1,2$, we have
(i)~$(p^c_0,q^c_0),(p^c_1,q^c_1),(p^c_2,q^c_2)$ in $M$ or (ii)~$(p^c_0,q^c_2),(p^c_1,q^c_1),(p^c_2,q^c_0)$ in $M$ or
(iii)~$(p^c_0,q^c_0),(p^c_1,q^c_2),(p^c_2,q^c_1)$ in $M$.

There are two unstable states for each level~2 gadget, i.e., either (ii) or (iii) above for the gadget on
$p^c_i,q^c_i$ for $i=0,1,2$. A level~2 gadget can be in either of these two unstable states in $M$---without loss of generality
assume that $M$ contains $(p^c_0,q^c_0),(p^c_1,q^c_2),(p^c_2,q^c_1)$ and $(p^c_3,q^c_5),(p^c_4,q^c_4),(p^c_5,q^c_3)$.
Observe that $p^c_2$ likes $y_j$ more than $q^c_1$ and similarly, $q^c_5$ likes $x_i$ more than $p^c_3$. Consider the following alternating path
$\rho$ with respect to~$M$: 
\[(q^c_2,p^c_1) \overset{(+,+)}{-}
(q^c_1,p^c_2) \overset{(+,-)}{-} 
(y_j,x_j) \overset{(-,+)}{-} 
(z_0,z_1) \overset{(+,-)}{-} 
(y_i,x_i) \overset{(-,+)}{-} 
(q^c_5,p^c_3) \overset{(+,+)}{-}
(q^c_3,p^c_5).\] 

Note that $M$ has to contain $(z_0,z_1)$ (by Lemma~\ref{lem:z_inside}).
Observe that $\rho$ is an alternating path in $G_M$ with {\em two} blocking edges $(p^c_1,q^c_1)$ and $(p^c_3,q^c_3)$.
This is a contradiction to $M$'s popularity (by Theorem~\ref{thm:char-popular}) and the lemma follows. \qed
\end{proof}

\begin{lemma}
  \label{lemma4}
Let $c = X_i \vee X_j \vee X_k$.
At most one of the level~1 gadgets corresponding to $X_i,X_j,X_k$ is in unstable state with respect to $M$.
\end{lemma}
\begin{proof}
  Suppose not. So at least two of the three level~1 gadgets corresponding to  $X_i,X_j,X_k$ are in unstable state with respect to~$M$. Assume without loss of generality
  that the gadgets corresponding to variables $X_i$ and $X_j$ are in unstable state. So the edges $(x_i,y'_i), (x'_i,y_i)$ are in $M$, similarly the edges
  $(x_j,y'_j), (x'_j,y_j)$ are in~$M$.

Recall that $a^c_5$ regards $y'_i$ as its second choice neighbor and $b^c_5$ regards $x'_j$ as its second choice neighbor. Similarly, $b'^c_5$ regards $x'_i$ as its second choice neighbor and $a'^c_5$ regards $y'_j$ as its second choice neighbor.

In the popular matching $M$, level~0 vertices are matched within their own gadget. Therefore, either $\{(a^c_5,b^c_5),(a^c_6,b^c_6)\} \subset M$ or
$\{(a^c_5,b^c_6),(a^c_6,b^c_5)\} \subset M$; similarly, $\{(a'^c_5,b'^c_5),(a'^c_6,b'^c_6)\} \subset M$ or $\{(a'^c_5,b'^c_6),(a'^c_6,b'^c_5)\} \subset M$. 
Thus, the following two observations clearly hold:
\begin{itemize}
\item either $a^c_5$ or $b^c_5$ is matched to its third choice neighbor;
\item either $a'^c_5$ or $b'^c_5$ is matched to its third choice neighbor.
\end{itemize}

Based on which of these vertices are matched to their third choice neighbors, we have four cases as shown below. 
Each of these 4 cases results in a forbidden alternating path/cycle (as given in Theorem~\ref{thm:char-popular}),
thus proving the lemma.

\smallskip

\noindent{\em Case 1.} The vertices $a^c_5$ and $a'^c_5$ are matched to their third choice neighbors. 
So $(a^c_5,b^c_6),(a^c_6,b^c_5)$ and $(a'^c_5,b'^c_6),(a'^c_6,b'^c_5)$ are in $M$.
Consider the following alternating path $\rho$ with respect to~$M$: 
\[(x'_i,y_i) \overset{(+,+)}{-}
(x_i,y'_i) \overset{(-,+)}{-}
(a^c_5,b^c_6) \overset{(-,+)}{-}
(z_0,z_1) \overset{(+,-)}{-}
(b'^c_6,a'^c_5) \overset{(+,-)}{-}
(y'_j,x_j) \overset{(+,+)}{-}
(y_j,x'_j).\]

Observe that $\rho$ is an alternating path in $G_M$ with two blocking edges $(x_i,y_i)$ and  $(x_j,y_j)$, a contradiction
to $M$'s popularity.

\smallskip

\noindent{\em Case 2.} The vertices $a^c_5$ and $b'^c_5$ are matched to their third choice neighbors. 
So $(a^c_5,b^c_6),(a^c_6,b^c_5)$ and $(a'^c_5,b'^c_5),(a'^c_6,b'^c_6)$ are in $M$.
Consider the following alternating cycle $C$ with respect to~$M$: 
\[(b^c_6,a^c_5) 
\overset{(+,-)}{-}
(y'_i,x_i) \overset{(+,+)}{-}
(y_i,x'_i) \overset{(-,+)}{-}
(b'^c_5,a'^c_5) \overset{(-,+)}{-}
(z_1,z_0) \overset{(+,-)}{-} 
(b^c_6,a^c_5).\]

Observe that $C$ is an alternating cycle in $G_M$ with a blocking edge $(x_i,y_i)$, a contradiction
to $M$'s popularity.

\smallskip

\noindent{\em Case 3.} The vertices $b^c_5$ and $a'^c_5$ are matched to their third choice neighbors. 
So $(a^c_5,b^c_5),(a^c_6,b^c_6)$ and $(a'^c_5,b'^c_6),(a'^c_6,b'^c_5)$ are in $M$.
Consider the following alternating cycle $C'$ with respect to~$M$:
\[  (a^c_5,b^c_5) \overset{(+,-)}{-}
(x'_j,y_j) \overset{(+,+)}{-}
(x_j,y'_j) \overset{(-,+)}{-}
(a'^c_5,b'^c_6) \overset{(-,+)}{-}
(z_0,z_1) \overset{(+,-)}{-}
(a^c_5,b^c_5).\]

Observe that $C'$ is an alternating cycle in $G_M$ with a blocking edge $(x_j,y_j)$, a contradiction
to $M$'s popularity.

\smallskip

\noindent{\em Case 4.} The vertices $b^c_5$ and $b'^c_5$ are matched to their third choice neighbors. 
So $(a^c_5,b^c_5),(a^c_6,b^c_6)$ and $(a'^c_5,b'^c_5),(a'^c_6,b'^c_6)$ are in $M$.
Consider the following alternating path $\rho'$ with respect to~$M$:
\[ (y'_i,x_i) \overset{(+,+)}{-}
(y_i,x'_i) \overset{(-,+)}{-}
(b'^c_5,a'^c_5) \overset{(-,+)}{-}
(z_0,z_1) \overset{(+,-)}{-}
(a^c_5,b^c_5) \overset{(+,-)}{-}
(x'_j,y_j) \overset{(+,+)}{-}
(x_j,y'_j).\]

Observe that $\rho'$ is an alternating path in $G_M$ with two blocking edges $(x_i,y_i)$ and  $(x_j,y_j)$, a contradiction
to $M$'s popularity. \qed
\end{proof}

Let $M$ be a popular matching in $G$.
It follows from Lemmas~\ref{lemma3} and~\ref{lemma4} that {\em exactly one} of the level~1 gadgets corresponding to the variables in $c$, for
every clause $c$ in $B$,
is in unstable state with respect to $M$. This allows us to set a 1-in-3 satisfying truth assignment to instance $B$. For each variable $X_i$ in $B$ do:

\begin{itemize}
\item if the gadget corresponding to $X_i$ is in {\em unstable} state then set $X_i = \mathsf{true}$ else set $X_i = \mathsf{false}$.
\end{itemize}

It follows from Lemmas~\ref{lemma3} and~\ref{lemma4} that
this is a 1-in-3 satisfying truth assignment for~$B$. We have thus shown the following result.

\begin{theorem}
  \label{thm1}
  If $G$ admits a popular matching then $B$ has a 1-in-3 satisfying truth assignment.
\end{theorem}  

\section{The converse}
\label{se:reverse-construction}
We will now show the converse of Theorem~\ref{thm1}, i.e., if $B$ has a 1-in-3 satisfying truth assignment $S$ then $G$ admits a popular matching. 
We will use $S$ to construct a popular matching $M$ in $G$ as follows. To begin with, $M = \emptyset$.

\paragraph{Level~1.}
For each variable $X_i$ do:
\begin{itemize}
 \item if $X_i$ is set to $\mathsf{true}$ in $S$ then add $(x_i,y'_i)$ and $(x'_i,y_i)$ to $M$;
 \item else  add $(x_i,y_i)$ and $(x'_i,y'_i)$ to $M$.
\end{itemize}

\noindent{\em Remark.} Note that the level~1 gadget of a variable set to true is in unstable state and the level~1 gadget of a variable
set to false is in stable state.

\smallskip

For each clause $c = X_i \vee X_j \vee X_k$, we know that exactly one of $X_i,X_j,X_k$ is set to $\mathsf{true}$ 
in $S$. Assume without loss of generality that $X_k = \mathsf{true}$ in~$S$. 
For the level~0, 2, and 3 gadgets corresponding to $c$, we do as follows:

\paragraph{Level~0.} Recall that there are {\em six} level~0 gadgets that correspond to $c$.
For the first 3 gadgets (these are on vertices $a^c_i,b^c_i$ for $i = 1,\ldots,6$) do:
\begin{itemize}
\item include $(a^c_1,b^c_2),(a^c_2,b^c_1)$ from the first gadget;
\item include $(a^c_3,b^c_3),(a^c_4,b^c_4)$ from the second gadget;
\item choose either $(a^c_5,b^c_5),(a^c_6,b^c_6)$ or $(a^c_5,b^c_6),(a^c_6,b^c_5)$ from the third gadget.
\end{itemize}

Observe that since the third variable $X_k$ of $c$ was set to be $\mathsf{true}$, cross edges are fixed in the first gadget 
(see Fig.~\ref{fig:0}), while the other stable matching (horizontal edges) is chosen in the second gadget. 

For the fourth and fifth gadgets, we will do exactly the opposite. Also, it will not matter which stable pair of edges is
chosen from the third and sixth gadgets.
So for the last 3 level~0 gadgets corresponding to $c$ (these are on vertices $a'^c_i,b'^c_i$ for $i = 1,\ldots,6$) do:
\begin{itemize}
\item include $(a'^c_1,b'^c_1),(a'^c_2,b'^c_2)$ from the fourth gadget;
\item include $(a'^c_3,b'^c_4),(a'^c_4,b'^c_3)$ from the fifth gadget.
\item choose either $(a'^c_5,b'^c_5),(a'^c_6,b'^c_6)$ or $(a'^c_5,b'^c_6),(a'^c_6,b'^c_5)$ from the sixth gadget.
\end{itemize}

\paragraph{Level~2.}Recall that there are {\em six} level~2 gadgets that correspond to $c$. For the first 3 gadgets (these are on vertices $p^c_i,q^c_i$ for $i = 0,\ldots,8$) do:
\begin{itemize}
\item include $(p^c_0,q^c_2),(p^c_1,q^c_1),(p^c_2,q^c_0)$ from the first gadget
\item include $(p^c_3,q^c_3),(p^c_4,q^c_5),(p^c_5,q^c_4)$ from the second gadget
\item include $(p^c_6,q^c_6),(p^c_7,q^c_7),(p^c_8,q^c_8)$ from the third gadget
\end{itemize}

In the first three gadgets, because $X_k = \mathsf{true}$, the third one is set to parallel edges, reaching 
the stable state, while the first one is blocked by the top horizontal edge and the second one is blocked by the middle horizontal edge.
Include isomorphic edges (to the above ones) from the last three level~2 gadgets corresponding to $c$, i.e., include 
$(p'^c_0,q'^c_2),(p'^c_1,q'^c_1),(p'^c_2,q'^c_0)$ from the fourth gadget, and so on. On this level, the last three gadgets mimic the matching 
edges from the first three gadgets, unlike in level~0. 
\paragraph{Level~3.} For the first level~3 gadget corresponding to $c$ do:
\begin{itemize}
\item include $(s^c_0,t^c_3),(s^c_1,t^c_1),(s^c_2,t^c_2),(s^c_3,t^c_0)$ in $M$.
\end{itemize}

Since the third variable in $c$ was set to be $\mathsf{true}$, the vertices $s_0^c$ and $t^c_0$ are matched to $t_3^c$ and $s_3^c$,
respectively---thus the bottom horizontal edge $(s_3^c,t_3^c)$ blocks~$M$.
Include isomorphic edges (to the above ones) for the second level~3 gadget corresponding to $c$, i.e., 
include $(s'^c_0,t'^c_3),(s'^c_1,t'^c_1),(s'^c_2,t'^c_2)$, $(s'^c_3,t'^c_0)$ in $M$. Once again, the second gadget mimics the matching edges 
on the first gadget.

\paragraph{$Z$-gadget and $D$-gadget.} Finally include the edges $(z_0,z_1),(z_2,z_3),(z_4,z_5)$ from the $Z$-gadget in $M$. By Lemma~\ref{lem:z_inside},
every popular matching in $G$ has to include these edges. Also include the edges $(d_0,d_1),(d_2,d_3)$ from the $D$-gadget in $M$.

\subsection{The popularity of $M$}
\label{se:reverse_popularity}

We will now prove the popularity of the above matching $M$ via the LP framework of popular matchings initiated in \cite{KMN09} for 
bipartite graphs. This framework generalizes to provide a sufficient condition for popularity in non-bipartite graphs~\cite{FKPZ18}.
This involves showing a witness $\vec{\alpha} \in \mathbb{R}^{n}$ such that $\vec{\alpha}$ is a {\em certificate} of $M$'s popularity.
In order to define the constraints that $\vec{\alpha}$ has to satisfy so as to certify $M$'s popularity,
let us define an edge weight function $w_M$ as follows.

For any edge $(u,v)$ in $G$ do:
\begin{itemize}
  \item if $(u,v)$ is labeled $(-,-)$ then set $w_M(u,v) = -2$;
  \item if $(u,v)$ is labeled $(+,+)$ then set $w_M(u,v) = 2$;
  \item else set  $w_M(u,v) = 0$. \ (So $w_M(e) = 0$ for all $e \in M$.)
\end{itemize}

Let $N$ be any perfect matching in $G$. It is easy to see from the definition of the edge weight function $w_M$ that 
$w_M(N) = \phi(N,M) - \phi(M,N)$.

Let the max-weight perfect fractional matching LP in the graph $G$ with edge weight function $w_M$ be our primal LP. This is \ref{LP0}
defined below. Here $\delta(u)$ denotes the set of edges incident to vertex~$u$.

\begin{linearprogram}
  {
    \label{LP0}
    \maximize{\sum_{e \in E} w_M(e)x_e}
  }
  \textstyle \sum_{e\in\delta(u)}x_e \ = \ 1  \ \ \forall\, u \in V\ & \ \ \ \ \ \ \text{and}\ \ \ \ \ \ \ x_e \ge 0 \ \ \forall\, e \in E\notag
\end{linearprogram}

If the optimal value of \ref{LP0} is at most 0 then $w_M(N) \le 0 $ for all perfect matchings $N$ in $G$, i.e.,
$\phi(N,M) \le \phi(M,N)$. Observe that this means $\phi(M',M) \le \phi(M,M')$ for {\em all} matchings $M'$ in $G$.
This is because $G$ is a complete graph on an even number of vertices, so for any matching $M'$, there is a perfect matching $N$ such that
$M' \subseteq N$. Thus $\phi(M',M) \le \phi(N,M) \le \phi(M,N) \le \phi(M,M')$.
Hence $M$ is a popular matching in~$G$.

Consider the LP that is dual to LP1. This is \ref{LP1} given below in variables $\alpha_u$, where $u \in V$.

\begin{linearprogram}
  {
    \label{LP1}
    \minimize{\sum_{u \in V} \alpha_u}
  }
  \textstyle \alpha_u + \alpha_v  \ & \ge \ \ w_M(u,v) \ \ \forall\, (u,v) \in E\notag
\end{linearprogram}

If we show a dual feasible solution $\vec{\alpha}$ such that $\sum_{u \in V}\alpha_u = 0$
then the primal optimal value is at most 0, i.e., $M$ is a popular matching. 
In order to prove the popularity of $M$, we define $\vec{\alpha}$ as follows. For $r \in \{1,\ldots,\kappa\}$ do: (recall that
$\kappa$ is the number of variables in the formula $B$)
\begin{itemize}
\item if $X_r$ was set to $\mathsf{true}$ then set $\alpha_{x_r} = \alpha_{y_r} = 1$ and $\alpha_{x'_r} = \alpha_{y'_r} = -1$;
\item else set $\alpha_{x_r} = \alpha_{y_r} = \alpha_{x'_r} = \alpha_{y'_r} = 0$.
\end{itemize}

Let clause $c = X_i \vee X_j \vee X_k$. 
Recall that we assumed that $X_i = X_j = \mathsf{false}$ and $X_k = \mathsf{true}$.
For the vertices in clauses corresponding to $c$, we will set $\alpha$-values as follows.
\begin{itemize}
\item For every level~0 vertex $v$ do: set $\alpha_v = 0$.
\item For the first three level~2 gadgets corresponding to $c$ do:
\begin{itemize}
\item set $\alpha_{p^c_0} = \alpha_{q^c_0} = 1$, $\alpha_{p^c_1} = 1, \alpha_{q^c_1} = -1$, and $\alpha_{p^c_2} = \alpha_{q^c_2} = -1$;
\item set $\alpha_{p^c_3} = -1, \alpha_{q^c_3} = 1$, $\alpha_{p^c_4} = \alpha_{q^c_4} = 1$, and $\alpha_{p^c_5} = \alpha_{q^c_5} = -1$;
\item set $\alpha_{p^c_6} = \alpha_{q^c_6} = \alpha_{p^c_7} = \alpha_{q^c_7} = \alpha_{p^c_8} = \alpha_{q^c_8} = 0$.
\end{itemize}
\end{itemize}

The setting of $\alpha$-values is analogous for vertices in the last three level~2 gadgets corresponding to $c$.
For the first level~3 gadget corresponding to $c$ do:
\begin{itemize}
\item set  $\alpha_{s^c_0} = \alpha_{t^c_0} = -1$, $\alpha_{s^c_1} = -1, \alpha_{t^c_1} = 1$, $\alpha_{s^c_2} = -1$, $\alpha_{t^c_2} = 1$, and 
$\alpha_{s^c_3} = \alpha_{t^c_3} = 1$.
\end{itemize}

The setting of $\alpha$-values is analogous for vertices in the other level~3 gadget corresponding to $c$.
For the $z$-vertices do: set $\alpha_u = 0$ for all $u \in \{z_0,\ldots,z_5\}$. For the $d$-vertices do:
\begin{itemize}
\item set $\alpha_{d_0} = \alpha_{d_2} = -1$ and $\alpha_{d_1} = \alpha_{d_3} = 1$.
\end{itemize}


\noindent{\bf Properties of $\vec{\alpha}$.}
For every $(u,v) \in M$, either $\alpha_u = \alpha_v = 0$ or $\{\alpha_u,\alpha_v\} = \{-1,1\}$; so $\alpha_u + \alpha_v = 0$.
Since $M$ is a perfect matching, we have $\sum_{u \in V}\alpha_u = 0$. The claims stated below 
show that $\vec{\alpha}$ is a feasible solution to \ref{LP1}. This will prove the popularity of $M$.

We need to show that every edge $(u,v)$ is {\em covered}, i.e., $\alpha_u + \alpha_v \ge w_M(u,v)$. We have already observed that
for any $(u,v) \in M$, $\alpha_u + \alpha_v = 0 = w_M(u,v)$.

\begin{new-claim}
\label{claim1}
  Let $(u,v)$ be an intra-gadget blocking edge to $M$. Then $\alpha_u + \alpha_v = 2 = w_M(u,v)$.
\end{new-claim}
\begin{proof}
Level~1 gadgets that correspond to variables set to $\mathsf{true}$ have blocking edges. More precisely, for every variable $X_k$ set to $\mathsf{true}$, $(x_k,y_k)$ is a blocking edge to $M$ and we have $\alpha_{x_k} = \alpha_{y_k} = 1$. Thus $\alpha_{x_k} + \alpha_{y_k} = 2 = w_M(x_k,y_k)$. Similarly, consider any level~2 or level~3 gadget that is in unstable state: such a gadget has a blocking edge within it, say $(p^c_0,q^c_0)$ or $(p^c_4,q^c_4)$ or $(s^c_3,t^c_3)$, and both endpoints of such an edge have their $\alpha$-values set to 1. For the $D$-gadget, $(d_1,d_3)$ is a blocking edge and we have $\alpha_{d_1} = \alpha_{d_3} = 1$. There are no blocking edges to $M$ in the $Z$-gadget or in a level~0 gadget. Thus all intra-gadget blocking edges are covered. \qed
\end{proof}

\begin{new-claim}
\label{claim2}
  Let $(u,v)$ be an intra-gadget edge that is non-blocking. Then $\alpha_u + \alpha_v \geq w_M(u,v)$.
\end{new-claim}
\begin{proof}
For any edge $(z_i,z_j)$ where $i,j \in \{0,1,\ldots,5\}$, we have $\alpha_{z_i} + \alpha_{z_j} = 0 = w_M(z_i,z_j)$.
Similarly, all edges within the $D$-gadget are covered. For any variable $X_i$ set to $\mathsf{false}$:
$\alpha_{x_i} + \alpha_{y'_i} = 0 = w_M(x_i,y'_i)$
and similarly, $\alpha_{x'_i} + \alpha_{y_i} = 0 = w_M(x'_i,y_i)$. For any variable $X_k$ set to $\mathsf{true}$:
$\alpha_{x'_k} + \alpha_{y'_k} = -2 = w_M(x'_k,y'_k)$.

For any $(a^c_i,b^c_i)$, we have $\alpha_{a^c_i} + \alpha_{b^c_i} = 0 = w_M(a^c_i,b^c_i)$. Similarly,
$\alpha_{a^c_{2i-1}} + \alpha_{b^c_{2i}} = 0 = w_M(a^c_{2i-1},b^c_{2i})$, also $\alpha_{a^c_{2i}} + \alpha_{b^c_{2i-1}} = 0 = w_M(a^c_{2i},b^c_{2i-1})$
for all $i$ and $c$.

We also have for all $c$: $\alpha_{p^c_1} + \alpha_{q^c_2} = 0 = w_M(p^c_1,q^c_2)$ while
$\alpha_{p^c_2} + \alpha_{q^c_1} = -2 = w_M(p^c_2,q^c_1)$ and $\alpha_{p^c_2} + \alpha_{q^c_2} = -2 = w_M(p^c_2,q^c_2)$.
It is similar for all other edges within level~2 gadgets and also for edges within level~3 gadgets.
Thus it is easy to see that for all intra-gadget non-blocking edges $(u,v)$, we have $\alpha_u + \alpha_v \ge w_M(u,v)$. \qed
\end{proof}

\begin{new-claim}
\label{claim3}
  Let $(u,v)$ be any inter-gadget edge. Then $\alpha_u + \alpha_v \geq w_M(u,v)$.
\end{new-claim}
\begin{proof}
  We show that no inter-gadget edge blocks~$M$. The vertices $z_0$ and $z_1$ prefer some neighbors in levels~0, 1, 2 to each other and the $\alpha$-value of each
  of these neighbors is either $0$ or $1$. In particular, $\alpha_{x_i} \ge 0$ and $\alpha_{y_i} \ge 0$, $\alpha_{p^c_1} \ge 0$ and $\alpha_{q^c_0} \ge 0$, and so on while 
  $\alpha_{a^c_i} = \alpha_{b^c_i} =  \alpha_{a'^c_i} = \alpha_{b'^c_i} = 0$ for all $i$ and $c$. Note that all these vertices prefer their partners in $M$ to $z_0$ or $z_1$; thus for any such edge $e$, we have
  $w_M(e) = 0$. Since $\alpha_{z_i} = 0$ for all $i$, the edges incident to $z_i$ are covered for all $i$.

Consider edges between a level~0 vertex and a level~1 vertex, such as $(a^c_1,y'_j)$ or $(b^c_1,x'_k)$: regarding the former edge, we have
$w_M(a^c_1,y'_j) = 0 = \alpha_{a^c_1} + \alpha_{y'_j}$ and for the latter edge, we have $w_M(b^c_1,x'_k) = -2 < -1 = \alpha_{b^c_1} + \alpha_{x'_k}$.
It can similarly be verified that every edge between a level~0 vertex and a level~1 vertex is covered.

Consider edges between a level~1 vertex and a level~2 vertex, such as $(p^c_2,y_j)$ or $(x_k,q^c_2)$: recall that $(p^c_0,q^c_2)$ and $(x_k,y'_k)$
are in $M$ and so $w_M(x_k,q^c_2) = 0$; we set $\alpha_{q^c_2} = -1$ and $\alpha_{x_k} = 1$, thus $\alpha_{x_k} + \alpha_{q^c_2} = w_M(x_k,q^c_2)$. 
We have $w_M(p^c_2,y_j) = -2$ since $(p^c_2,q^c_0)$ and $(x_j,y_j)$ are in $M$ and so this edge is covered.
It can similarly be verified that every edge between a level~1 vertex and a level~2 vertex is covered.

Consider edges between a level~2 vertex and a level~3 vertex, such as those incident to $s^c_0$ or $t^c_0$: we have
$w_M(s^c_0,q^c_0) = w_M(s^c_0,q^c_3) = 0$ and $\alpha_{s^c_0} = -1$ while $\alpha_{q^c_0} = \alpha_{q^c_3} = 1$.
Similarly, $w_M(p^c_7,t^c_0) = w_M(p^c_4,t^c_0) = -2$ and so these edges are covered.
It is analogous with edges incident to $s'^c_0$ or $t'^c_0$.

Consider any edge $e$ whose one endpoint is in the $D$-gadget and the other endpoint is outside the $D$-gadget. It is easy to see that
$w_M(e) = -2$, hence this edge is covered. Similarly, inter-gadget edges between levels 0 and 2, levels 0 and 3, and levels 1 and 3 all have weight $-2$ and hence they are covered. \qed
\end{proof}

Thus we have shown the following theorem.

\begin{theorem}
  \label{thm2}
  If $B$ has a 1-in-3 satisfying truth assignment then $G$ admits a popular matching.
\end{theorem}

Theorem~\ref{thm:main} stated in Section~\ref{sec:intro} follows from Theorems~\ref{thm1} and \ref{thm2}.
Thus the popular matching problem in a roommates instance on $n$ vertices with complete preference lists is $\NP$-complete for even $n$.

\medskip
\medskip

\noindent{\bf Acknowledgments.} Thanks to Chien-Chung Huang for asking us about the complexity of the popular roommates problem
with complete preference lists. We are grateful to the reviewers for their helpful comments and suggestions that improved the presentation of the paper.

\end{document}